\documentclass[11pt]{article}%
\usepackage{amsmath}
\usepackage{amssymb}
\usepackage{amsfonts}

\usepackage{cite}
\usepackage{graphicx}
\usepackage{float}
\usepackage{longtable}
\usepackage[utf8]{inputenc}
\usepackage{hyperref}

\usepackage{algorithmicx}
\usepackage[ruled,vlined]{algorithm2e}
\SetAlgoCaptionSeparator{ }

\usepackage{fullpage}

\newcommand{\fref}[1]{Fig.~\ref{#1}}

\newcommand{\tref}[1]{Table~\ref{#1}}
\newcommand{\sref}[1]{Section~\ref{#1}}

\providecommand{\U}[1]{\protect\rule{.1in}{.1in}}
\newtheorem{theorem}{Theorem}

\newtheorem{definition}{Definition}

\newtheorem{lemma}{Lemma}

\newtheorem{problem}{Problem}
\newtheorem{proposition}{Proposition}

\newenvironment{proof}[1][Proof]{\textbf{#1.} }{\ \rule{0.5em}{0.5em}}

\begin{document}

\title{Theory of Noise-Scaled Stability Bounds and Entanglement Rate Maximization in the Quantum Internet}
\author{Laszlo Gyongyosi\thanks{School of Electronics and Computer Science, University of Southampton, Southampton SO17 1BJ, U.K., and Department of Networked Systems and Services, Budapest University of Technology and Economics, 1117 Budapest, Hungary, and MTA-BME Information Systems Research Group, Hungarian Academy of Sciences, 1051 Budapest, Hungary.}
\and Sandor Imre\thanks{Department of Networked Systems and Services, Budapest University of Technology and Economics, 1117 Budapest, Hungary.}
}

\date{}

\maketitle
\begin{abstract}
Crucial problems of the quantum Internet are the derivation of stability properties of quantum repeaters and theory of entanglement rate maximization in an entangled network structure. The stability property of a quantum repeater entails that all incoming density matrices can be swapped with a target density matrix. The strong stability of a quantum repeater implies stable entanglement swapping with the boundness of stored density matrices in the quantum memory and the boundness of delays. Here, a theoretical framework of noise-scaled stability analysis and entanglement rate maximization is conceived for the quantum Internet. We define the term of entanglement swapping set that models the status of quantum memory of a quantum repeater with the stored density matrices. We determine the optimal entanglement swapping method that maximizes the entanglement rate of the quantum repeaters at the different entanglement swapping sets as function of the noise of the local memory and local operations. We prove the stability properties for non-complete entanglement swapping sets, complete entanglement swapping sets and perfect entanglement swapping sets. We prove the entanglement rates for the different entanglement swapping sets and noise levels. The results can be applied to the experimental quantum Internet.
\end{abstract}

\section{Introduction}
\label{sec1}
The quantum Internet allows legal parties to perform networking based on the fundamentals of quantum mechanics \cite{ref1,ref2,ref3,ref11,ref13a,ref13,ref18,refn7,puj1,puj2,pqkd1,np1}. The connections in the quantum Internet are formulated by a set of quantum repeaters and the legal parties have access to large-scale quantum devices \cite{ref5,ref6,ref7,qmemuj,ref4} such as quantum computers \cite{qc1,qc2,qc3,qc4,qc5,qc6,qcadd1,qcadd2,qcadd3,qcadd4,shor1,refibm}. Quantum repeaters are physical devices with quantum memory and internal procedures \cite{ref5,ref6,ref7,ref11,ref13a,ref13,ref8,ref9,ref10,add1,add2,add3,refqirg,ref18,ref19,ref20,ref21,add4,refn7,refn5,refn3,sat,telep,refn1,refn2,refn4,refn6}. An aim of the quantum repeaters is to generate the entangled network structure of the quantum Internet via entanglement distribution \cite{ref23,ref24,ref25,ref26,ref27, nadd1,nadd2,nadd3,nadd4,nadd5,nadd6,nadd7,kris1,kris2}. The entangled network structure can then serve as the core network of a global-scale quantum communication network with unlimited distances (due to the attributes of the entanglement distribution procedure). Quantum repeaters share entangled states over shorter distances; the distance can be extended by the entanglement swapping operation in the quantum repeaters \cite{refn7,ref1,ref5,ref6,ref7,ref13,ref13a}. The swapping operation takes an incoming density matrix and an outgoing density matrix; both density matrices are stored in the local quantum memory of the quantum repeater \cite{ref45,ref46,ref47,ref48,ref49,ref50,ref51,ref52,ref53,ref54,ref55,ref56,ref57,ref58,ref60,ref61,ref62
}. The incoming density matrix is half of an entangled state such that the other half is stored in the distant source node, while the outgoing density matrix is half of an entangled state such that the other half is stored in the distant target node. The entanglement swapping operation, applied on the incoming and outgoing density matrices in a particular quantum repeater, entangles the distant source and target quantum nodes. Crucial problems here are the size and delay bounds connected to the local quantum memory of a quantum repeater and the optimization of the swapping procedure such that the entanglement rate of the quantum repeater (outgoing entanglement throughput measured in entangled density matrices per a preset time unit) is maximal. These questions lead us to the necessity of strictly defining the fundamental stability and performance criterions \cite{ref36,ref37,ref38,ref39,ref40,ref41,ref42,ref43,ref44} of quantum repeaters in the quantum Internet. 

Here, a theoretical framework of noise-scaled stability analysis and entanglement rate maximization is defined for the quantum Internet. By definition, the stability of a quantum repeater can be weak or strong. The strong stability implies weak stability, by some fundamentals of queueing theory \cite{refL1,refL2,refL3,refL4,refD1}. Weak stability of a quantum repeater entails that all incoming density matrices can be swapped with a target density matrix. Strong stability of a quantum repeater further guarantees the boundness of the number of stored density matrices in the local quantum memory. The defined system model of a quantum repeater assumes that the incoming density matrices are stored in the local quantum memory of the quantum repeater. The stored density matrices formulate the set of incoming density matrices (input set). The quantum memory also consists of a separate set for the outgoing density matrices (output set). Without loss of generality, the cardinality of the input set (number of stored density matrices) is higher than the cardinality of the output set. Specifically, the cardinality of the input set is determined by the entanglement throughput of the input connections, while the cardinality of the output set equals the number of output connections. Therefore, if in a given swapping period, the number of incoming density matrices exceeds the cardinality of the output set, then several incoming density matrices must be stored in the input set (Note: The logical model of the storage mechanisms of entanglement swapping in a quantum repeater is therefore analogous to the logical model of an input-queued switch architecture \cite{refL1,refL2,refL3}.). The aim of entanglement swapping is to select the density matrices from the input and output sets, such that the outgoing entanglement rate of the quantum repeater is maximized; this also entails the boundness of delays. The maximization procedure characterizes the problem of optimal entanglement swapping in the quantum repeaters. 

Finding the optimal entanglement swapping means determining the entanglement swapping between the incoming and outgoing density matrices that maximizes the outgoing entanglement rate of the quantum repeaters. The problem of entanglement rate maximization must be solved for a particular noise level in the quantum repeater and with the presence of various entanglement swapping sets. The noise level in the proposed model is analogous to the lost density matrices in the quantum repeater due to imperfections in the local operations and errors in the quantum memory units. The entanglement swapping sets are logical sets that represent the actual state of the quantum memory in the quantum repeater. The entanglement swapping sets are formulated by the set of received density matrices stored in the local quantum memory and the set of outgoing density matrices, which are also stored in the local quantum memory. Each incoming and outgoing density matrix represent half of an entangled system, such that the other half of an incoming density matrix is stored in the distant source quantum repeater, while the other half of an outgoing density matrix is stored in the distant target quantum repeater. The aim of determining the optimal entanglement swapping method is to apply the local entanglement swapping operation on the set of incoming and outgoing density matrices such that the outgoing entanglement rate of the quantum repeater is maximized at a particular noise level. As we prove, the entanglement rate maximization procedure depends on the type of entanglement swapping sets formulated by the stored density matrices in the quantum memory. We define the logical types of the entanglement swapping sets and characterize the main attributes of the swapping sets. We present the efficiency of the entanglement swapping procedure as a function of the local noise and its impacts on the entanglement rate. We prove that the entanglement swapping sets can be defined as a function of the noise, which allows us to define noise-scaled entanglement swapping and noise-scaled entanglement rate maximization. The proposed theoretical framework utilizes the fundamentals of queueing theory, such as the Lyapunov methodology \cite{refL1}, which is an analytical tool used to assess the performance of queueing systems \cite{refL1,refL2,refL3,refL4,refD1,refs3}, and defines a fusion of queueing theory with quantum Shannon theory \cite{ref4,ref22,ref28,ref29,ref30,ref32,ref33,ref34,ref35} and the theory of quantum Internet.

The novel contributions of our manuscript are as follows:
\begin{enumerate}
\item  We define a theoretical framework of noise-scaled entanglement rate maximization for the quantum Internet. 

\item  We determine the optimal entanglement swapping method that maximizes the entanglement rate of a quantum repeater at the different entanglement swapping sets as a function of the noise level of the local memory and local operations.

\item  We prove the stability properties for non-complete entanglement swapping sets, complete entanglement swapping sets and perfect entanglement swapping sets. 

\item  We prove the entanglement rate of a quantum repeater as a function of the entanglement swapping sets and the noise level.
\end{enumerate}

This paper is organized as follows. In \sref{sec2}, the preliminary definitions are discussed. \sref{sec3} proposes the noise-scaled stability analysis. In \sref{sec4}, the noise-scaled entanglement rate maximization is defined. \sref{sec5} provides a performance evaluation. Finally, \sref{sec6} concludes the results. Supplemental information is included in the Appendix.

\section{System Model and Problem Statement}
\label{sec2}
\subsection{System Model}
Let $V$ refer to the nodes of an entangled quantum network $N$, which consists of a transmitter node $A\in V$, a receiver node $B\in V$, and quantum repeater nodes $R_{i} \in V$, $i=1,\ldots ,q$. Let $E=\left\{E_{j} \right\}$, $j=1,\ldots ,m$ refer to a set of edges (an edge refers to an entangled connection in a graph representation) between the nodes of $V$, where each $E_{j} $ identifies an ${\rm L}_{l} $-level entanglement, $l=1,\ldots ,r$, between quantum nodes $x_{j} $ and $y_{j} $ of edge $E_{j} $, respectively. Let $N=\left(V,{\rm {\mathcal S}}\right)$ be an actual quantum network with $\left|V\right|$ nodes and a set ${\rm {\mathcal S}}$ of entangled connections. An ${\rm L}_{l} $-level, $l=1,\ldots ,r$, entangled connection $E_{{\rm L}_{l} } \left(x,y\right)$, refers to the shared entanglement between a source node $x$ and a target node $y$, with hop-distance 
\begin{equation} \label{ZEqnNum202142} 
d\left(x,y\right)_{{\rm L}_{l} } =2^{l-1} ,                                               
\end{equation} 
since the entanglement swapping (extension) procedure doubles the span of the entangled pair in each step. This architecture is also referred to as the doubling architecture \cite{ref1,ref5,ref6,ref7}. 

For a particular ${\rm L}_{l} $-level entangled connection $E_{{\rm L}_{l} } \left(x,y\right)$ with hop-distance \eqref{ZEqnNum202142}, there are $d\left(x,y\right)_{{\rm L}_{l} } -1$ intermediate nodes between the quantum nodes $x$ and $y$.

\fref{fig1} depicts a quantum Internet scenario with an intermediate quantum repeater $R_{j} $. The aim of the quantum repeater is to generate long-distance entangled connections between the distant quantum repeaters. The long-distance entangled connections are generated by the $U_{S} $ entanglement swapping operation applied in $R_{j} $. The quantum repeater must manage several different connections with heterogeneous entanglement rates. The density matrices are stored in the local quantum memory of the quantum repeater. The aim is to find an entanglement swapping in $R_{j} $ that maximizes the entanglement rate of the quantum repeater.

\begin{center}
\begin{figure*}[!htbp]
\begin{center}
\includegraphics[angle = 0,width=1\linewidth]{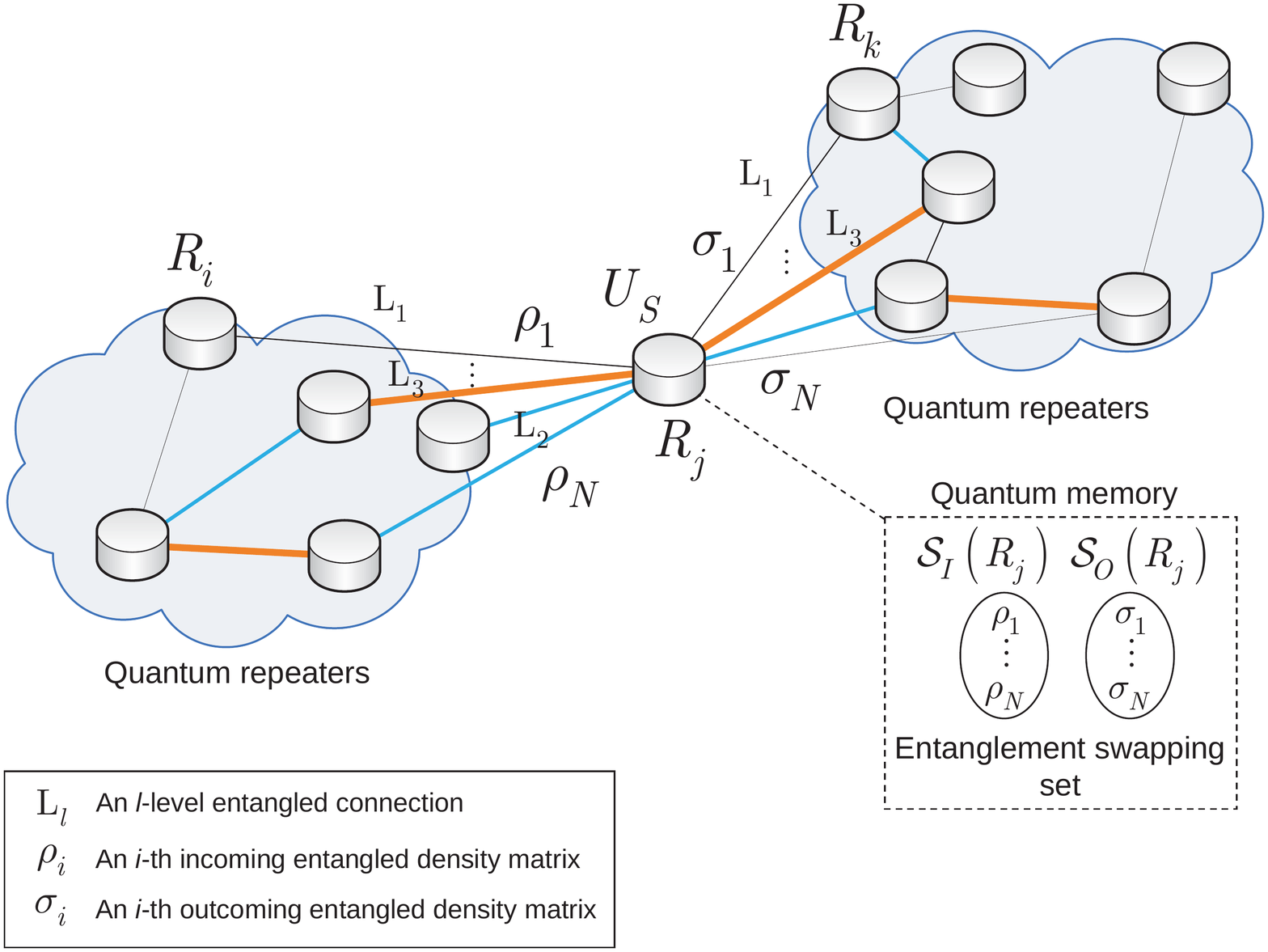}
\caption{The problem of entanglement swapping in quantum repeater $R_{j} $ with $N$ input and $N$ output connections in a quantum Internet scenario. Quantum repeater $R_{j} $ stores an $\rho _{i} $ incoming entangled density matrix from the $i$-th input (the other half of $\rho _{i} $ is shared with a source quantum repeater $R_{i} $) and the $\sigma _{k} $ outgoing entangled density matrix (the other half of $\sigma _{k} $ is shared with a target quantum repeater $R_{k} $) in its local quantum memory. The $U_{S} $ entanglement swapping operation in $R_{j} $ generates long-distance entangled connections between the distant quantum nodes. The incoming and outgoing density matrices formulate sets ${\rm {\mathcal S}}_{I} \left(R_{j} \right)$ and ${\rm {\mathcal S}}_{O} \left(R_{j} \right)$ together formulate the entanglement swapping set. The aim of the optimization procedure is to determine the optimal entanglement swapping to maximize the outgoing entanglement rate of $R_{j} $.)} 
 \label{fig1}
 \end{center}
\end{figure*}
\end{center}

\subsubsection{Entanglement Fidelity}
The aim of the entanglement distribution procedure is to establish a $d$-dimensional entangled system between the distant points $A$ and $B$, through the intermediate quantum repeater nodes. Let $d=2$, and let ${\left| \beta _{00}  \right\rangle} $ be the target entangled system $A$ and $B$, ${\left| \beta _{00}  \right\rangle} =\frac{1}{\sqrt{2} } \left({\left| 00 \right\rangle} +{\left| 11 \right\rangle} \right),$ subject to be generated. At a particular density $\sigma $ generated between $A$ and $B$, the fidelity of $\sigma $ is evaluated as
\begin{equation} \label{ZEqnNum728497} 
F=\left\langle  {{\beta }_{00}} | \sigma |{{\beta }_{00}} \right\rangle ,
\end{equation} 
Without loss of generality, an aim of a practical entanglement distribution is to reach $F\ge 0.98$ in \eqref{ZEqnNum728497} for a given $\sigma $ \cite{ref1,ref2,ref3,ref4,ref5,ref6,ref7,ref8}.
\subsubsection{Entanglement Purification and Entanglement Throughput}
Entanglement purification \cite{purifuj1,purifuj2,purifuj3} is a probabilistic procedure that creates a higher fidelity entangled system from two low-fidelity Bell states. The entanglement purification procedure yields a Bell state with an increased entanglement fidelity $F'$, 
\begin{equation} \label{3)} 
F_{in} <F' \le 1,                                                 
\end{equation} 
where $F_{in} $ is the fidelity of the imperfect input Bell pairs. The purification requires the use of two-way classical communications \cite{ref1,ref2,ref3,ref4,ref5,ref6,ref7,ref8}.

Let $B_{F} (E_{{\rm L}_{l} }^{i})$ refer to the entanglement throughput of a given ${\rm L}_{l} $ entangled connection $E_{{\rm L}_{l} }^{i} $ measured in the number of $d$-dimensional entangled states established over $E_{{\rm L}_{l} }^{i} $ per sec at a particular fidelity $F$ (dimension of a qubit system is $d=2$) \cite{ref1,ref2,ref3,ref4,ref5,ref6,ref7,ref8}. 

For any entangled connection $E_{{\rm L}_{l} }^{i} $, a condition $c$ should be satisfied, as
\begin{equation} \label{ZEqnNum801212}
c:{{B}_{F}}( E_{{{\text{L}}_{l}}}^{i})\ge {B}_{F}^{\text{*}}( E_{{{\text{L}}_{l}}}^{i}),\text{ for }\forall i,
\end{equation} 
where ${{B}}_{F}^{\text{*}}( E_{{{\text{L}}_{l}}}^{i})$ is a critical lower bound on the entanglement throughput at a particular fidelity $F$ of a given $E_{{{\text{L}}_{l}}}^{i}$, i.e., ${{B}_{F}}( E_{{{\text{L}}_{l}}}^{i})$ of a particular $E_{{{\text{L}}_{l}}}^{i}$ has to be at least ${B}_{F}^{\text{*}}( E_{{{\text{L}}_{l}}}^{i})$.

\subsection{Definitions}
Some preliminary definitions for the proposed model are as follows.
\begin{definition}(Incoming and outgoing density matrix). In a $j$-th quantum repeater $R_{j} $, an $\rho $ incoming density matrix is half of an entangled state $\left| {{\beta }_{00}} \right\rangle =\tfrac{1}{\sqrt{2}}\left( \left| 00 \right\rangle +\left| 11 \right\rangle  \right)$ received from a previous neighbor node $R_{j-1} $. The $\sigma $ outgoing density matrix in $R_{j} $ is half of an entangled state ${\left| \beta _{00}  \right\rangle} $ shared with a next neighbor node $R_{j+1} $.
\end{definition}
\begin{definition}(Entanglement Swapping Operation). The $U_{S} $ entanglement swapping operation is a local transformation in a $j$-th quantum repeater $R_{j}$  that swaps an incoming density matrix $\rho $ with an outgoing density matrix $\sigma $ and measures the density matrices to entangle the distant source and target nodes $R_{j-1} $ and $R_{j+1} $. 
\end{definition}
\begin{definition}(Entanglement Swapping Period). Let $C$ be a cycle with time $t_{C} ={1\mathord{\left/ {\vphantom {1 f_{C} }} \right. \kern-\nulldelimiterspace} f_{C} } $ determined by the $o_{C} $ oscillator in node $R_{j} $, where $f_{C} $ is the frequency of $o_{C} $. Then, let $\pi _{S} $ be an entanglement swapping period in which the set ${\rm {\mathcal S}}_{I} \left(R_{j} \right)=\bigcup _{i}\rho _{i}  $ of incoming density matrices is swapped via $U_{S} $ with the set ${\rm {\mathcal S}}_{O} \left(R_{j} \right)=\bigcup _{i}\sigma _{i}  $ of outgoing density matrices, defined as $\pi _{S} =xt_{C} $, where $x$ is the number of $C$.  
\end{definition}
\begin{definition}(Complete and Non-Complete Swapping Sets). Set ${\rm {\mathcal S}}_{I} \left(R_{j} \right)$ formulates a complete set ${\rm {\mathcal S}}_{I}^{*} \left(R_{j} \right)$ if set ${\rm {\mathcal S}}_{I} \left(R_{j} \right)$ contains all the $Q=\sum _{i=1}^{N}\left|B_{i} \right| $ incoming density matrices per $\pi _{S} $ that is received by $R_{j} $ during a swapping period, where $N$ is the number of input entangled connections of $R_{j} $ and $\left|B_{i} \right|$ is the number of incoming densities of the $i$-th input connection per $\pi _{S} $; thus, ${\rm {\mathcal S}}_{I} \left(R_{j} \right)=\bigcup _{i=1}^{Q}\rho _{i}  $ and $\left|{\rm {\mathcal S}}_{I} \left(R_{j} \right)\right|=Q$. Set ${\rm {\mathcal S}}_{O} \left(R_{j} \right)$ formulates a complete set ${\rm {\mathcal S}}_{O}^{*} \left(R_{j} \right)$ if ${\rm {\mathcal S}}_{O} \left(R_{j} \right)$ contains all the $N$ outgoing density matrices that are shared by $R_{j} $ during a swapping period $\pi _{S} $; thus, ${\rm {\mathcal S}}_{O} \left(R_{j} \right)=\bigcup _{i=1}^{N}\sigma _{i}  $ and $\left|{\rm {\mathcal S}}_{O} \left(R_{j} \right)\right|=N$. 

Let ${\rm {\mathcal S}}\left(R_{j} \right)$ be an entanglement swapping set of $R_{j} $, defined as 
\begin{equation} \label{1)} 
{\rm {\mathcal S}}\left(R_{j} \right)={\rm {\mathcal S}}_{I} \left(R_{j} \right)\bigcup {\rm {\mathcal S}}_{O} \left(R_{j} \right) .   
\end{equation} 
Then, ${\rm {\mathcal S}}\left(R_{j} \right)$ is a complete swapping ${\rm {\mathcal S}}^{*} \left(R_{j} \right)$ set, if 
\begin{equation} \label{2)} 
{\rm {\mathcal S}}^{*} \left(R_{j} \right)={\rm {\mathcal S}}_{I}^{*} \left(R_{j} \right)\bigcup {\rm {\mathcal S}}_{O}^{*} \left(R_{j} \right) ,   
\end{equation} 
with cardinality
\begin{equation} \label{3)} 
\left|{\rm {\mathcal S}}^{*} \left(R_{j} \right)\right|=Q+N.   
\end{equation} 
Otherwise, ${\rm {\mathcal S}}\left(R_{j} \right)$ formulates a non-complete swapping set ${\rm {\mathcal S}}\left(R_{j} \right)$, with cardinality
\begin{equation} \label{4)} 
\left|{\rm {\mathcal S}}\left(R_{j} \right)\right|<Q+N.  
\end{equation} 
\end{definition}
\begin{definition}(Perfect Swapping Sets). A complete swapping set ${\rm {\mathcal S}}^{*} \left(R_{j} \right)$ is a perfect swapping set 
\begin{equation}\label{5)} 
\hat{{\rm {\mathcal S}}}\left(R_{j} \right)=\hat{{\rm {\mathcal S}}}_{I} \left(R_{j} \right)\bigcup \hat{{\rm {\mathcal S}}}_{O} \left(R_{j} \right)
\end{equation}
at a given $\pi _{S} $, if 
\begin{equation} \label{6)} 
\left|\hat{{\rm {\mathcal S}}}\left(R_{j} \right)\right|=N+N 
\end{equation} 
holds for the cardinality.
\end{definition}
\begin{definition}(Coincidence set). In a given $\pi _{S} $, the coincidence set ${\rm {\mathcal S}}_{R_{j} }^{\left(\pi _{S} \right)} \left(\left(R_{i} ,\sigma _{k} \right)\right)$ is a subset of incoming density matrices in ${\rm {\mathcal S}}_{I} \left(R_{j} \right)$ of $R_{j} $ received from $R_{i} $ that requires the outgoing density matrix $\sigma _{k} $ from ${\rm {\mathcal S}}_{O} \left(R_{j} \right)$ for the entanglement swapping. The cardinality of the coincidence set is
\begin{equation} \label{7)} 
Z_{R_{j} }^{\left(\pi _{S} \right)} \left(\left(R_{i} ,\sigma _{k} \right)\right)=\left|{\rm {\mathcal S}}_{R_{j} }^{\left(\pi _{S} \right)} \left(\left(R_{i} ,\sigma _{k} \right)\right)\right|.   
\end{equation} 
\end{definition}
\begin{definition}(Coincidence set increment in an entanglement swapping period) Let $\left|B\left(R_{i} \left(\pi _{S} \right),\sigma _{k} \right)\right|$ refer to the number of density matrices arriving from $R_{i} $ for swapping with $\sigma _{k} $ at $\pi _{S} $. This means the increment of the $Z_{R_{j} }^{\left(\pi '_{S} \right)} \left(\left(R_{i} ,\sigma _{k} \right)\right)$ cardinality of the coincidence set is
\begin{equation} \label{8)} 
Z_{R_{j} }^{\left(\pi '_{S} \right)} \left(\left(R_{i} ,\sigma _{k} \right)\right)=Z_{R_{j} }^{\left(\pi _{S} \right)} \left(\left(R_{i} ,\sigma _{k} \right)\right)+\left|B\left(R_{i} \left(\pi _{S} \right),\sigma _{k} \right)\right|,  
\end{equation} 
where $\pi '_{S} $ is the next entanglement swapping period. The derivations assume that an incoming density matrix $\rho $ chooses a particular output density matrix $\sigma $ for the entanglement swapping with probability $\Pr \left(\rho ,\sigma \right)=x\ge 0$ (Bernoulli i.i.d.).
\end{definition}
\begin{definition}(Incoming and outgoing entanglement rate) Let $\left|B_{R_{i} } \left(\pi _{S} \right)\right|$ be the incoming entanglement rate of $R_{j} $ per a given $\pi _{S} $, defined as
\begin{equation} \label{9)} 
\left|B_{R_{j} } \left(\pi _{S} \right)\right|=\sum _{i,k}\left|B\left(R_{i} \left(\pi _{S} \right),\sigma _{k} \right)\right| ,  
\end{equation} 
where $\left|B\left(R_{i} \left(\pi _{S} \right),\sigma _{k} \right)\right|$ refers to the number of density matrices arriving from $R_{i} $ for swapping with $\sigma _{k} $ per $\pi _{S} $.

Then, at a given $\left|B_{R_{i} } \left(\pi _{S} \right)\right|$, the $\left|B'_{R_{j} } \left(\pi _{S} \right)\right|$, the outgoing entanglement rate of $R_{j} $ is defined as
\begin{equation} \label{10)} 
\left|B'_{R_{j} } \left(\pi _{S} \right)\right|=\left(1-{\tfrac{L}{N}} \right){\tfrac{1}{1+D\left(\pi _{S} \right)}} \left(\left|B_{R_{j} } \left(\pi _{S} \right)\right|\right), 
\end{equation} 
where $L$ is the loss, $0<L\le N$, and $D\left(\pi _{S} \right)$ is a delay measured in entanglement swapping periods caused by the optimal entanglement swapping at a particular entanglement swapping set.
\end{definition}
\begin{definition}(Swapping constraint). In a given $\pi _{S} $, each incoming density in ${\rm {\mathcal S}}_{I} \left(R_{j} \right)$ can be swapped with at most one outgoing density, and only one outgoing density is available in ${\rm {\mathcal S}}_{O} \left(R_{j} \right)$ for each outgoing entangled connection.
\end{definition}
\begin{definition}(Weak stable (stable) and strongly stable entanglement swapping). Weak stability (stability) of a quantum repeater $R_{j}$ entails that all incoming density matrices can be swapped with a target density matrix in $R_{j}$. 
A $\zeta \left(\pi _{S} \right)$ entanglement swapping in $R_{j}$ is weak stable (stable) if, for every $\varepsilon >0$, there exists a $B>0$, such that
\begin{equation} \label{ZEqnNum415455} 
\mathop{\lim }\limits_{\pi _{S} \to \infty } \Pr \left(\left|{\rm {\mathcal S}}_{I}^{\left(\pi _{S} \right)} \left(R_{j} \right)\right|>B\right)<\varepsilon ,  
\end{equation} 
where ${\rm {\mathcal S}}_{I}^{\left(\pi _{S} \right)} \left(R_{j} \right)$ is the set of incoming densities of $R_{j} $ at $\pi _{S} $, while $\left|{\rm {\mathcal S}}_{I}^{\left(\pi _{S} \right)} \left(R_{j} \right)\right|$ is the cardinality of the set. 

For a strongly stable entanglement swapping in $R_{j}$, the weak stability is satisifed and the cardinality of ${\rm {\mathcal S}}_{I}^{\left(\pi _{S} \right)} \left(R_{j} \right)$ is bounded. A $\zeta \left(\pi _{S} \right)$ entanglement swapping in $R_{j} $ is strongly stable if
\begin{equation} \label{ZEqnNum908737} 
\mathop{\lim }\limits_{\pi _{S} \to \infty } \sup {\mathbb{E}}\left(\left|{\rm {\mathcal S}}_{I}^{\left(\pi _{S} \right)} \left(R_{j} \right)\right|\right)<\infty .  
\end{equation} 
\end{definition}

\subsubsection{Noise-Scaled Entanglement Swapping Sets}
\begin{proposition}(Noise Scaled Swapping Sets). Let $\gamma $ be a noise coefficient that models the noise of the local quantum memory and the local operations, $0\le \gamma \le 1$. For $\gamma =0$, the swapping set at a given $\pi _{S} $ is complete swapping set ${\rm {\mathcal S}}^{*} \left(R_{j} \right)$, while for any $\gamma >0$, the swapping set is non-complete swapping set ${\rm {\mathcal S}}\left(R_{j} \right)$ at a given $\pi _{S} $.
\end{proposition}

In realistic situations, $\gamma $ corresponds to the noises and imperfections of the physical devices and physical-layer operations (quantum operations, realization of quantum gates, storage errors, losses from local physical devices, optical losses, etc) in the quantum repeater that lead to the loss of density matrices. For further details on the physical-layer aspects of repeater-assisted quantum communications in an experimental quantum Internet setting, we suggest \cite{ref13}.

\fref{fig2} illustrates the perfect swapping set, complete swapping set and non-complete swapping set. For both sets, the incoming densities are stored in incoming set ${\rm {\mathcal S}}_{I} \left(R_{j} \right)$. Its cardinality depends on the incoming entanglement throughputs of the incoming connections. The outgoing set ${\rm {\mathcal S}}_{O} \left(R_{j} \right)$ is a collection of outgoing density matrices. The outgoing matrix is half of an entangled state and the other half is shared with a distant target node.

The input set ${\rm {\mathcal S}}_{I} \left(R_{j} \right)$ and output se ${\rm {\mathcal S}}_{O} \left(R_{j} \right)$ of $R_{j} $ consist of the incoming and outgoing density matrices. For a non-complete entanglement swapping set, the noise is non-zero; therefore, loss is present in the quantum memory. As a convention of our model (see the swapping constraint in Definition 9), any density matrix loss is modeled as a ``double loss'' that affects both sets ${\rm {\mathcal S}}_{I} \left(R_{j} \right)$ and ${\rm {\mathcal S}}_{O} \left(R_{j} \right)$. Because of a loss, the $U_{S} $ swapping operation cannot be performed on the incoming and outgoing density matrices.

\begin{center}
\begin{figure*}[!htbp]
\begin{center}
\includegraphics[angle = 0,width=1\linewidth]{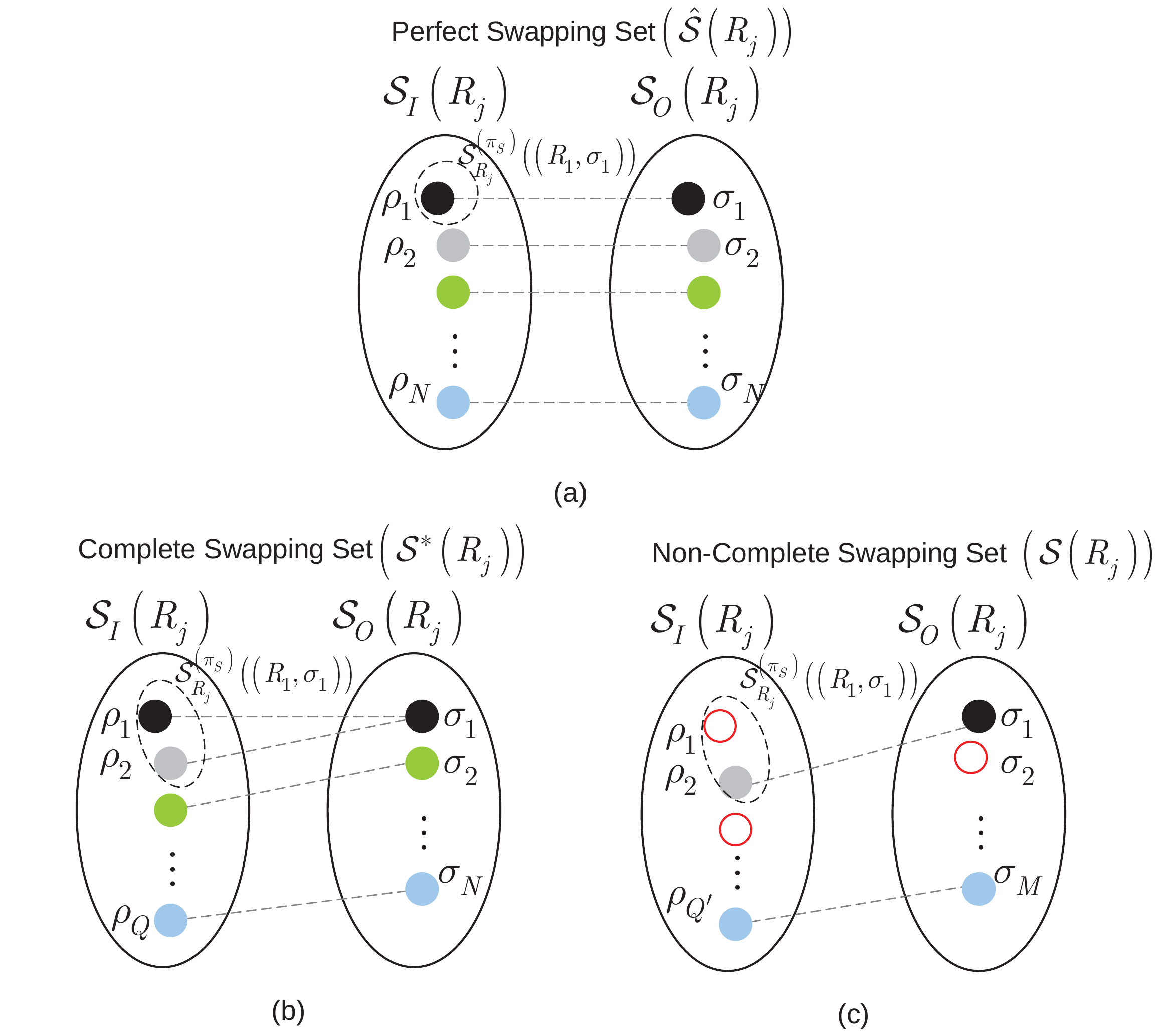}
\caption{Logical types of entanglement swapping sets at a given $\pi _{S} $. (a). For a perfect entanglement swapping set $\hat{{\rm {\mathcal S}}}\left(R_{j} \right)$, the cardinalities of the input and output sets are $\left|\hat{{\rm {\mathcal S}}}_{I} \left(R_{j} \right)\right|=N$ and $\left|\hat{{\rm {\mathcal S}}}_{O} \left(R_{j} \right)\right|=N$; thus, $\left|\hat{{\rm {\mathcal S}}}\left(R_{j} \right)\right|=N+N=2N$. The $Z_{R_{j} }^{\left(\pi _{S} \right)} \left(\left(R_{i} ,\sigma _{k} \right)\right)$ cardinality of the coincidence sets ${\rm {\mathcal S}}_{R_{j} }^{\left(\pi _{S} \right)} \left(\left(R_{i} ,\sigma _{k} \right)\right)$, $i=1,\ldots ,N$, $k=1,\ldots ,N$, is $Z_{R_{j} }^{\left(\pi _{S} \right)} \left(\left(R_{i} ,\sigma _{k} \right)\right)=1$. (b). For a complete entanglement swapping set ${\rm {\mathcal S}}^{*} \left(R_{j} \right)$, $\left|{\rm {\mathcal S}}_{I}^{*} \left(R_{j} \right)\right|=Q>N$ and $\left|{\rm {\mathcal S}}_{O}^{*} \left(R_{j} \right)\right|=N$; thus, $\left|{\rm {\mathcal S}}^{*} \left(R_{j} \right)\right|=Q+N$. The cardinality of the coincidence sets ${\rm {\mathcal S}}_{R_{j} }^{\left(\pi _{S} \right)} \left(\left(R_{i} ,\sigma _{k} \right)\right)$, $i=1,\ldots ,N$ and $k=1,\ldots ,N$ is $Z_{R_{j} }^{\left(\pi _{S} \right)} \left(\left(R_{i} ,\sigma _{k} \right)\right)\ge 1$. (c). For a non-complete entanglement swapping set ${\rm {\mathcal S}}\left(R_{j} \right)$, some densities are randomly lost due to noise (depicted by empty dots) in ${\rm {\mathcal S}}_{I} \left(R_{j} \right)$, leading to $\left|{\rm {\mathcal S}}^{*} \left(R_{j} \right)\right|=Q'+M$, where $Q'\le Q$, and $M=N-L$, where $L$ is the number of lost densities at a given $\pi _{S} $. Since for each ${\rm {\mathcal S}}_{R_{j} }^{\left(\pi _{S} \right)} \left(\left(R_{i} ,\sigma _{k} \right)\right)$ an output $\sigma _{k} $ is associated at a given $\pi _{S} $, in a density loss in a given coincidence set also causes a decrease in the cardinality of the output set ${\rm {\mathcal S}}_{O} \left(R_{j} \right)$ (due to the swapping constraint); thus, $\left|{\rm {\mathcal S}}_{O} \left(R_{j} \right)\right|=M$. The cardinality of the coincidence sets ${\rm {\mathcal S}}_{R_{j} }^{\left(\pi _{S} \right)} \left(\left(R_{i} ,\sigma _{k} \right)\right)$, $i=1,\ldots ,N$ and $k=1,\ldots ,M$ is $Z_{R_{j} }^{\left(\pi _{S} \right)} \left(\left(R_{i} ,\sigma _{k} \right)\right)\ge 0$.} 
 \label{fig2}
 \end{center}
\end{figure*}
\end{center}
 
\subsection{Problem Statement}
The problem formulation for the noise-scaled entanglement rate maximization is given in Problems 1-3. 
\begin{problem}
Determine the entanglement swapping method that maximizes the entanglement rate of a quantum repeater at the different entanglement swapping sets as a function of the noise level of the local memory and local operations.
\end{problem}
\begin{problem}
Prove the stability for non-complete entanglement swapping sets, complete entanglement swapping sets and perfect entanglement swapping sets.
\end{problem}
\begin{problem}
Determine the outgoing entanglement rate of a quantum repeater as a function of the entanglement swapping sets and the noise level.
\end{problem}
\begin{problem}
Define the optimal entanglement swapping period length as a function of the noise level at the different entanglement swapping sets.
\end{problem}
The resolutions of Problems 1-4 are proposed in Proposition 1 and Theorems 1-4.
\section{Entanglement Swapping Stability at Swapping Sets}
\label{sec3}
This section presents the stability analysis of the quantum repeaters for the different entanglement swapping sets.
\begin{proposition}(Noise-scaled weight coefficient). Let $\omega \left(\gamma \left(\pi _{S} \right)\right)$ be a weight at a non-complete swapping set ${\rm {\mathcal S}}\left(R_{j} \right)$ at $\gamma \left(\pi _{S} \right)>0$, where $\gamma \left(\pi _{S} \right)$ is the noise at a given $\pi _{S} $. For a ${\rm {\mathcal S}}^{*} \left(R_{j} \right)$ complete swapping set, $\gamma \left(\pi _{S} \right)=0$, and $\omega \left(\gamma \left(\pi _{S} \right)=0\right)=\omega ^{*} \left(\pi _{S} \right)$ is a maximized weight at $\pi _{S} $. For any non-complete set ${\rm {\mathcal S}}\left(R_{j} \right)$, $\omega \left(\gamma \left(\pi _{S} \right)\right)\ge \omega ^{*} \left(\pi _{S} \right)-f\left(\gamma \left(\pi _{S} \right)\right)$, where $f\left(\cdot \right)$ is a sub-linear function, and $\omega \left(\gamma \left(\pi _{S} \right)\right)<\omega ^{*} \left(\pi _{S} \right)$. At a perfect swapping set $\hat{{\rm {\mathcal S}}}\left(R_{j} \right)$, the weight is $\hat{\omega }\left(\pi _{S} \right)\le \omega ^{*} \left(\pi _{S} \right)$.
\end{proposition}
\begin{proof}
At a given $\pi _{S} $, let $\zeta _{ik} \left(\rho _{A} ,\sigma _{k} \right)$ be constant, with respect to the swapping constraint, defined as
\begin{equation} \label{ZEqnNum888718} 
\zeta _{ik} \left(\rho _{A} ,\sigma _{k} \right)=\left\{\begin{array}{l} {1,{\rm \; if\; }\rho _{A} \in {\rm {\mathcal S}}_{R_{j} }^{\left(\pi _{S} \right)} \left(\left(R_{i} ,\sigma _{k} \right)\right){\rm \; }} \\ {0,{\rm \; otherwise,}} \end{array}\right.  
\end{equation} 
thus $\zeta _{ik} \left(\rho _{A} ,\sigma _{k} \right)=1$, if incoming density $\rho _{A} $ is selected from set ${\rm {\mathcal S}}_{R_{j} }^{\left(\pi _{S} \right)} \left(\left(R_{i} ,\sigma _{k} \right)\right)$ for the swapping with $\sigma _{k} $, and $0$ otherwise. The aim is therefore to construct a 
\begin{equation} \label{ZEqnNum544183} 
\zeta \left(\pi _{S} \right)=\zeta _{ik} \left(\rho _{A} ,\sigma _{k} \right)_{i\le N, k\le N}  
\end{equation} 
feasible entanglement swapping method for all input and output neighbors of $R_{j} $, for all $\pi _{S} $ entanglement swapping periods.

Then, from $Z_{R_{j} }^{\left(\pi _{S} \right)} \left(\left(R_{i} ,\sigma _{k} \right)\right)$ (see Definition 7) and \eqref{ZEqnNum888718}, a $\omega \left(\pi _{S} \right)$ weight coefficient can be defined for a given entanglement swapping $\zeta \left(\pi _{S} \right)$ at a given $\pi _{S} $, as
\begin{equation} \label{15)} 
\begin{split}
   \omega \left( {{\pi }_{S}} \right)&=\sum\limits_{i,k}{{{\zeta }_{ik}}\left( {{\rho }_{A}},{{\sigma }_{k}} \right)Z_{{{R}_{j}}}^{\left( {{\pi }_{S}} \right)}\left( \left( {{R}_{i}},{{\sigma }_{k}} \right) \right)} \\ 
 & =\left\langle \zeta \left( {{\pi }_{S}} \right), {{Z}_{{{R}_{j}}}}\left( {{\pi }_{S}} \right) \right\rangle ,  
\end{split}
\end{equation} 
where $\left\langle \cdot \right\rangle $ is the inner product, $Z_{R_{j} } \left(\pi _{S} \right)$ is a matrix of all coincidence set cardinalities for all input and output connections at $\pi _{S} $, defined as
\begin{equation} \label{ZEqnNum354275} 
Z_{R_{j} } \left(\pi _{S} \right)=Z_{R_{j} }^{\left(\pi _{S} \right)} \left(\left(R_{i} ,\sigma _{k} \right)\right)_{i\le N,k\le N} ,   
\end{equation} 
while $\zeta \left(\pi _{S} \right)$ is as given in \eqref{ZEqnNum544183}.

For a perfect and complete entanglement swapping set, at $\gamma \left(\pi _{S} \right)=0$, let $\zeta ^{*} \left(\pi _{S} \right)$ refer to the entanglement swapping, with $\omega \left(\gamma \left(\pi _{S} \right)=0\right)=\omega ^{*} \left(\pi _{S} \right)$, where  $\omega ^{*} \left(\pi _{S} \right)$ is the maximized weight coefficient defined as
\begin{equation} \label{ZEqnNum770876} 
\omega ^{*} \left(\gamma \left(\pi _{S} \right)=0\right)=\mathop{\max }\limits_{\zeta ^{*} \left(\pi _{S} \right)} \left\langle \zeta ^{*} \left(\pi _{S} \right),Z_{R_{j} } \left(\pi _{S} \right)\right\rangle ,   
\end{equation} 
where $\zeta ^{*} \left(\pi _{S} \right)$ is an optimal entanglement swapping method at $\gamma \left(\pi _{S} \right)=0$ (in general not unique, by theory) with the maximized weight. By some fundamental theory \cite{refL1,refL2,refL3,refD1}, it can be verified that for a non-complete set with entanglement swapping $\chi \left(\pi _{S} \right)$ at $\gamma \left(\pi _{S} \right)>0$, defined as
\begin{equation} \label{18)} 
\chi \left(\pi _{S} \right)=\left\{\left(\chi _{ik} \left(\rho _{A}, \sigma _{k} \right)\right)_{x}, x=Mi+k, i=0,\ldots , M-1, k=0, \ldots , M-1\right\} 
\end{equation} 
with norm a $\left|\chi \left(\pi _{S} \right)\right|$ \cite{refL1,refL3},  defined as
\begin{equation} \label{19)} 
\left|\chi \left(\pi _{S} \right)\right|=\mathop{\max }\limits_{k=0,\ldots ,M-1} \left(\sum _{x=0}^{M-1}\left(\chi _{ik} \left(\rho _{A} ,\sigma _{k} \right)\right)_{Mx+k} ,\sum _{y=0}^{M-1}\left(\chi _{ik} \left(\rho _{A} ,\sigma _{k} \right)\right)_{Mk+y}   \right),   
\end{equation} 
with $\left|\chi \left(\pi _{S} \right)\right|\le 1$, the relation
\begin{equation} \label{20)} 
\left\langle \zeta ^{*} \left(\pi _{S} \right),Z_{R_{j} } \left(\pi _{S} \right)\right\rangle -\left\langle \chi \left(\pi _{S} \right),Z_{R_{j} } \left(\pi _{S} \right)\right\rangle \ge 0 
\end{equation} 
holds for the weights.

Then, let ${\rm {\mathcal L}}\left(Z_{R_{j} } \left(\pi _{S} \right)\right)$ be a Lyapunov function \cite{refL1,refL2,refL3,refD1}of $Z_{R_{j} } \left(\pi _{S} \right)$ (see \eqref{ZEqnNum354275}) as
\begin{equation} \label{ZEqnNum792527} 
{\rm {\mathcal L}}\left(Z_{R_{j} } \left(\pi _{S} \right)\right)=\sum _{i,k}\left(Z_{R_{j} }^{\left(\pi _{S} \right)} \left(\left(R_{i} ,\sigma _{k} \right)\right)\right)^{2}.  
\end{equation} 

Then it can be verified \cite{refL1,refL2,refL3,refD1} that
\begin{equation} \label{ZEqnNum314655} 
{\mathbb{E}}\left({\rm {\mathcal L}}\left(Z_{R_{j} } \left(\pi '_{S} \right)\right)-\left. {\rm {\mathcal L}}\left(Z_{R_{j} } \left(\pi _{S} \right)\right)\right|Z_{R_{j} } \left(\pi _{S} \right)\right)\le -\varepsilon \left|Z_{R_{j} } \left(\pi _{S} \right)\right|,  
\end{equation} 
holds, as $Z_{R_{j} } \left(\pi _{S} \right)$ is sufficiently large \cite{refL1,refL2,refL3,refD1}, where $\varepsilon >0$. Since \eqref{ZEqnNum314655} is analogous to the condition on strong stability given in \eqref{ZEqnNum908737}, it follows that as \eqref{ZEqnNum770876} holds for all $\pi _{S} $ entanglement swapping periods, the entanglement swapping at $\gamma \left(\pi _{S} \right)=0$ in $R_{j} $ is a strongly stable entanglement swapping with maximized weight coefficients for all periods. 

Since for any complete and perfect entanglement swapping set, the noise is zero, the $\hat{\omega }\left(\pi _{S} \right)$ weight coefficient of at a perfect swapping set $\hat{{\rm {\mathcal S}}}\left(R_{j} \right)$ is also a maximized weight with $f\left(\gamma \left(\pi _{S} \right)\right)=0$ as
\begin{equation} \label{23)} 
\hat{\omega }\left(\pi _{S} \right)\le \omega ^{*} \left(\pi _{S} \right).   
\end{equation} 
As the noise is nonzero, $\gamma \left(\pi _{S} \right)>0$, an $\zeta \left(\pi _{S} \right)\ne \zeta ^{*} \left(\pi _{S} \right)$ entanglement swapping cannot reach the $\omega ^{*} \left(\pi _{S} \right)$ maximized weight coefficient \eqref{ZEqnNum770876}, thus
\begin{equation} \label{24)} 
\omega (\gamma \left(\pi _{S} \right))=\mathop{\max }\limits_{\zeta \left(\pi _{S} \right)} \left\langle \zeta \left(\pi _{S} \right),Z_{R_{j} } \left(\pi _{S} \right)\right\rangle <\omega ^{*} \left(\pi _{S} \right), 
\end{equation} 
where the non-zero noise $\gamma \left(\pi _{S} \right)$ decreases $\omega ^{*} \left(\pi _{S} \right)$ to $\omega \left(\gamma \left(\pi _{S} \right)\right)<\omega ^{*} \left(\pi _{S} \right)$  as
\begin{equation} \label{ZEqnNum577058} 
\omega \left(\gamma \left(\pi _{S} \right)\right)\ge \omega ^{*} \left(\gamma \left(\pi _{S} \right)=0\right)-f\left(\gamma \left(\pi _{S} \right)\right),  
\end{equation} 
where $f\left(\cdot \right)$ is a sub-linear function, as
\begin{equation} \label{26)} 
0\le f\left(\gamma \left(\pi _{S} \right)\right)<c\left(\gamma \left(\pi _{S} \right)\right),  
\end{equation} 
and
\begin{equation} \label{27)} 
\mathop{\lim }\limits_{\gamma \left(\pi _{S} \right)\to \infty } {\tfrac{f\left(\gamma \left(\pi _{S} \right)\right)}{\gamma \left(\pi _{S} \right)}} =0,  
\end{equation} 
where $\gamma \left(\pi _{S} \right)\ge \gamma \left(\pi _{S} =0\right)$, and $c>0$ is a  constant. 

For a non-complete set, \eqref{ZEqnNum314655} can be rewritten as
\begin{equation} \label{28)} 
{\mathbb{E}}\left({\rm {\mathcal L}}\left(Z_{R_{j} } \left(\pi '_{S} \right)\right)-\left. {\rm {\mathcal L}}\left(Z_{R_{j} } \left(\pi _{S} \right)\right)\right|Z_{R_{j} } \left(\pi _{S} \right)\right)\le -\varepsilon ,  
\end{equation} 
where $\varepsilon >0$, which represents the stability condition given in \eqref{ZEqnNum415455}. As follows, an entanglement swapping at a non-complete entanglement swapping set is stable. 
\end{proof}

These statements are further improved in Theorems 1 and 2, respectively.
\subsection{Non-Complete Swapping Sets}
\begin{theorem}
(Noise-scaled stability at non-complete swapping sets) An $\zeta \left(\pi _{S} \right)$ entanglement swapping at $\gamma \left(\pi _{S} \right)>0$ is stable for any non-complete entanglement swapping set ${\rm {\mathcal S}}\left(R_{j} \right)$.
\end{theorem}
\begin{proof}
Let $\gamma \left(\pi _{S} \right)>0$ be the noise at a given $\pi _{S} $, and let $\zeta \left(\pi _{S} \right)$ be the actual entanglement swapping at any non-complete entanglement swapping set ${\rm {\mathcal S}}\left(R_{j} \right)$. Using the formalism of \cite{refD1}, let ${\rm {\mathcal L}}\left(X\right)$ refer to a Lyapunov function of a $M\times M$ size matrix $X$, defined as
\begin{equation} \label{29)} 
{\rm {\mathcal L}}\left(X\right)=\sum _{i,k}\left(x_{ik} \right)^{2}  ,   
\end{equation} 
where $x_{ik} $ is the $\left(i,k\right)$-th element of $X$. 

Let $C_{1} $ and $C_{2} $ constants, $C_{1} >0$, $C_{2} >0$. Then, an $\zeta \left(\pi _{S} \right)$ entanglement swapping with $\gamma \left(\pi _{S} \right)>0$ is stable if only
\begin{equation} \label{ZEqnNum985408} 
{\mathbb{E}}\left(\left. \Delta _{{\rm {\mathcal L}}} \right|Z_{R_{j} } \left(\pi _{S} \right)\right)\le -C_{1} \omega \left(\gamma \left(\pi _{S} \right)=0\right), 
\end{equation} 
where $\Delta _{{\rm {\mathcal L}}} $ is a difference of the Lyapunov functions ${\rm {\mathcal L}}\left(Z_{R_{j} } \left(\pi '_{S} \right)\right)$ and ${\rm {\mathcal L}}\left(Z_{R_{j} } \left(\pi _{S} \right)\right)$, where $\pi '_{S} $ is a next entanglement swapping period, defined as 
\begin{equation} \label{31)} 
\Delta _{{\rm {\mathcal L}}} ={\rm {\mathcal L}}\left(Z_{R_{j} } \left(\pi '_{S} \right)\right)-{\rm {\mathcal L}}\left(Z_{R_{j} } \left(\pi _{S} \right)\right),  
\end{equation} 
and 
\begin{equation} \label{32)} 
\omega \left(\gamma \left(\pi _{S} \right)=0\right)=\omega ^{*} \left(\pi _{S} \right)\ge C_{2} ,   
\end{equation} 
by theory \cite{refL1,refL2,refL3,refD1}. 

To verify \eqref{ZEqnNum985408}, first $\Delta _{{\rm {\mathcal L}}} $ is rewritten via \eqref{ZEqnNum792527} as
\begin{equation} \label{ZEqnNum432521} 
\begin{split}
   {{\Delta }_{\mathcal{L}}}&=\sum\limits_{i,k}{{{\left( Z_{{{R}_{j}}}^{\left( {{{{\pi }'_{S}}}} \right)}\left( \left( {{R}_{i}},{{\sigma }_{k}} \right) \right) \right)}^{2}}-{{\left( Z_{{{R}_{j}}}^{\left( {{\pi }_{S}} \right)}\left( \left( {{R}_{i}},{{\sigma }_{k}} \right) \right) \right)}^{2}}} \\ 
 & =\sum\limits_{i,k}{\left( Z_{{{R}_{j}}}^{\left( {{{{\pi }'_{S}}}} \right)}\left( \left( {{R}_{i}},{{\sigma }_{k}} \right) \right)-Z_{{{R}_{j}}}^{\left( {{\pi }_{S}} \right)}\left( \left( {{R}_{i}},{{\sigma }_{k}} \right) \right) \right)\left( Z_{{{R}_{j}}}^{\left( {{{{\pi }'_{S}}}} \right)}\left( \left( {{R}_{i}},{{\sigma }_{k}} \right) \right)+Z_{{{R}_{j}}}^{\left( {{\pi }_{S}} \right)}\left( \left( {{R}_{i}},{{\sigma }_{k}} \right) \right) \right),}  
\end{split}
\end{equation} 
where $Z_{R_{j} } \left(\pi '_{S} \right)$ can be evaluated as
\begin{equation} \label{ZEqnNum742078} 
\begin{split}
   {{Z}_{{{R}_{j}}}}\left( {{{{\pi }'_{S}}}} \right)&=\left( {{Z}_{{{R}_{j}}}}\left( {{\pi }_{S}} \right)-{{\zeta }_{ik}}\left( {{\rho }_{A}},{{\sigma }_{k}} \right) \right)+\left| \bar{B}\left( {{R}_{i}}\left( {{{{\pi }'_{S}}}} \right),{{\sigma }_{k}} \right) \right| \\ 
 & \le \max \left( \left( {{Z}_{{{R}_{j}}}}\left( {{\pi }_{S}} \right)-{{\zeta }_{ik}}\left( {{\rho }_{A}},{{\sigma }_{k}} \right) \right)+\left| \bar{B}\left( {{R}_{i}}\left( {{{{\pi }'_{S}}}} \right),{{\sigma }_{k}} \right) \right|,1 \right),  
\end{split}
\end{equation} 
where $\left|\bar{B}\left(R_{i} \left(\pi '_{S} \right),\sigma _{k} \right)\right|\le 1$ is the normalized number of arrival density matrices from $R_{i} $ for swapping with $\sigma _{k} $ at a next entanglement swapping period $\pi '_{S} $, defined as
\begin{equation} \label{ZEqnNum733839} 
\left|\bar{B}\left(R_{i} \left(\pi '_{S} \right),\sigma _{k} \right)\right|={\tfrac{\left|B\left(R_{i} \left(\pi '_{S} \right),\sigma _{k} \right)\right|}{\left|B_{R_{j} } \left(\pi _{S} \right)\right|}} ,  
\end{equation} 
where $\left|B_{R_{j} } \left(\pi _{S} \right)\right|=\sum _{i,k}\left|B\left(R_{i} \left(\pi _{S} \right),\sigma _{k} \right)\right| $ is a total number of incoming density matrices of $R_{j} $ from the $N$ quantum repeaters.

Using \eqref{ZEqnNum742078}, the result in \eqref{ZEqnNum432521} can be rewritten
\begin{equation} \label{36)} 
\begin{split}
   {{\Delta }_{\mathcal{L}}}&\le \sum\limits_{i,k}{\left( \left| \bar{B}\left( {{R}_{i}}\left( {{{{\pi }'_{S}}}} \right),{{\sigma }_{k}} \right) \right|-{{\zeta }_{ik}}\left( {{\rho }_{A}},{{\sigma }_{k}} \right) \right)\left( \left( 2Z_{{{R}_{j}}}^{\left( {{\pi }_{S}} \right)}\left( \left( {{R}_{i}},{{\sigma }_{k}} \right) \right)+1 \right)+1 \right)} \\ 
 & \le \sum\limits_{i,k}{\left( \left| \bar{B}\left( {{R}_{i}}\left( {{{{\pi }'_{S}}}} \right),{{\sigma }_{k}} \right) \right|-{{\zeta }_{ik}}\left( {{\rho }_{A}},{{\sigma }_{k}} \right) \right)\left( \left( 2Z_{{{R}_{j}}}^{\left( {{\pi }_{S}} \right)}\left( \left( {{R}_{i}},{{\sigma }_{k}} \right) \right) \right) \right)+2{{M}^{2}},}  
\end{split}
\end{equation} 
thus \eqref{ZEqnNum985408} can be rewritten as

\begin{equation} \label{ZEqnNum677350}
\begin{split}
   \mathbb{E}\left( \left. {{\Delta }_{\mathcal{L}}} \right|{{Z}_{{{R}_{j}}}}\left( {{\pi }_{S}} \right) \right)&\le 2\sum\limits_{i,k}{Z_{{{R}_{j}}}^{\left( {{\pi }_{S}} \right)}\left( \left( {{R}_{i}},{{\sigma }_{k}} \right) \right)\left( \mathbb{E}\left( \left| \bar{B}\left( {{R}_{i}}\left( {{\pi }_{S}} \right),{{\sigma }_{k}} \right) \right| \right)-\left. {{\zeta }_{ik}}\left( {{\rho }_{A}},{{\sigma }_{k}} \right) \right|{{Z}_{{{R}_{j}}}}\left( {{\pi }_{S}} \right) \right)+2{{M}^{2}}} \\ 
 & =2\sum\limits_{i,k}{Z_{{{R}_{j}}}^{\left( {{\pi }_{S}} \right)}\left( \left( {{R}_{i}},{{\sigma }_{k}} \right) \right)\left( \left| \bar{B}\left( {{R}_{i}}\left( {{\pi }_{S}} \right),{{\sigma }_{k}} \right) \right|-{{\zeta }_{ik}}\left( {{\rho }_{A}},{{\sigma }_{k}} \right) \right)+2{{\left( N-L \right)}^{2}}},  
\end{split}
\end{equation}
where ${\mathbb{E}}\left(\left|\bar{B}\left(R_{i} \left(\pi _{S} \right),\sigma _{k} \right)\right|\right)$ is the expected normalized number of density matrices arrive from $R_{i} $ for swapping with $\sigma _{k} $ at $\pi _{S} $. 

Let $\omega \left(\pi _{S} \left(\gamma \left(\pi _{S} \right)>0\right)\right)$ be the weight coefficient of the $\zeta \left(\pi _{S} \left(\gamma \left(\pi _{S} \right)>0\right)\right)$ actual entanglement swapping at the noise $\gamma \left(\pi _{S} \right)>0$ at a given $\pi _{S} $, as
\begin{equation} \label{38)} 
\omega \left(\pi _{S} \left(\gamma \left(\pi _{S} \right)>0\right)\right)=\left\langle \zeta \left(\pi _{S} \left(\gamma \left(\pi _{S} \right)>0\right)\right),Z_{R_{j} } \left(\pi _{S} \right)\right\rangle , 
\end{equation} 
and let $\alpha _{ik} $ be defined as
\begin{equation} \label{39)} 
\alpha _{ik} =Z_{R_{j} }^{\left(\pi _{S} \right)} \left(\left(R_{i} ,\sigma _{k} \right)\right)\bar{B}\left(R_{i} \left(\pi _{S} \right),\sigma _{k} \right).  
\end{equation} 
Then, by some fundamental theory \cite{refL1,refL2,refL3,refD1},
\begin{equation} \label{ZEqnNum101849} 
\begin{split}
   \sum\limits_{i,k}{{{\alpha }_{ik}}}&\le \sum\limits_{z}{{{\nu }_{z}}\left\langle {{\zeta }_{z}}\left( {{\pi }_{S}}\left( \gamma \left( {{\pi }_{S}} \right)>0 \right) \right),{{Z}_{{{R}_{j}}}}\left( {{\pi }_{S}} \right) \right\rangle } \\ 
 & =\sum\limits_{z}{{{\nu }_{z}}{{{{\omega }'_{z}}}}\left( {{\pi }_{S}}\left( \gamma \left( {{\pi }_{S}} \right)>0 \right) \right),}  
\end{split} 
\end{equation} 
where $\nu _{z} \ge 0$ is a constant, and $\zeta _{z} \left(\pi _{S} \left(\gamma \left(\pi _{S} \right)>0\right)\right)$ is a $z$-th entanglement swapping at a noise $\gamma \left(\pi _{S} \right)>0$, while $\omega '_{z} \left(\pi _{S} \left(\gamma \left(\pi _{S} \right)>0\right)\right)$ is the weight of $\zeta _{z} \left(\pi _{S} \left(\gamma \left(\pi _{S} \right)>0\right)\right)$.

Then, using \eqref{ZEqnNum101849}, ${\mathbb{E}}\left(\left. \Delta _{{\rm {\mathcal L}}} \right|Z_{R_{j} } \left(\pi _{S} \right)\right)$ from \eqref{ZEqnNum677350} can be rewritten as
 \begin{equation} \label{ZEqnNum588923}
\begin{split}
  & \mathbb{E}\left( \left. {{\Delta }_{\mathcal{L}}} \right|{{Z}_{{{R}_{j}}}}\left( {{\pi }_{S}} \right) \right) \\ 
 & \le 2\left( \sum\limits_{z}{{{\nu }_{z}}{{{{\omega }'_{z}}}}\left( {{\pi }_{S}}\left( \gamma \left( {{\pi }_{S}} \right)>0 \right) \right)-\omega \left( {{\pi }_{S}}\left( \gamma \left( {{\pi }_{S}} \right)>0 \right) \right)} \right)+2{{\left( N-L \right)}^{2}} \\ 
 & =2\left( \sum\limits_{z}{\begin{split}
  & {{\nu }_{z}}{{{{\omega }'_{z}}}}\left( {{\pi }_{S}}\left( \gamma \left( {{\pi }_{S}} \right)>0 \right) \right)-\omega \left( {{\pi }_{S}}\left( \gamma \left( {{\pi }_{S}} \right)=0 \right) \right) \\ 
 & +\omega \left( {{\pi }_{S}}\left( \gamma \left( {{\pi }_{S}} \right)=0 \right) \right)-\omega \left( {{\pi }_{S}}\left( \gamma \left( {{\pi }_{S}} \right)>0 \right) \right) \\ 
\end{split}} \right)+2{{\left( N-L \right)}^{2}} \\ 
 & \le 2\left( \sum\limits_{z}{{{\nu }_{z}}-1} \right)\omega \left( {{\pi }_{S}}\left( \gamma \left( {{\pi }_{S}} \right)=0 \right) \right)+2f\left( \gamma \left( {{\pi }_{S}} \right) \right)+2{{\left( N-L \right)}^{2}} \\ 
 & =-2{{C}_{1}}\omega \left( {{\pi }_{S}}\left( \gamma \left( {{\pi }_{S}} \right)=0 \right) \right)+2f\left( \gamma \left( {{\pi }_{S}} \right) \right)+2{{\left( N-L \right)}^{2}},  
\end{split}
\end{equation} 
where $C_{1} $ is set as \cite{refD1}
\begin{equation} \label{ZEqnNum400853} 
C_{1} =1-\sum _{z}\nu _{z}  .  
\end{equation} 
Therefore, there as $C_{2} $ is selected such that 
\begin{equation} \label{43)} 
0<C_{2} \le \omega \left(\pi _{S} \left(\gamma \left(\pi _{S} \right)=0\right)\right),  
\end{equation} 
then
\begin{equation} \label{ZEqnNum846609} 
{\mathbb{E}}\left(\left. \Delta _{{\rm {\mathcal L}}} \right|Z_{R_{j} } \left(\pi _{S} \right)\right)\le -C_{1} \omega \left(\pi _{S} \left(\gamma \left(\pi _{S} \right)=0\right)\right) 
\end{equation} 
holds, by theory \cite{refL1,refL2,refL3,refD1}, since \eqref{ZEqnNum846609} can be rewritten as
\begin{equation} \label{45)} 
{\mathbb{E}}\left(\left. \Delta _{{\rm {\mathcal L}}} \right|Z_{R_{j} } \left(\pi _{S} \right)\right)\le -\varepsilon ,  
\end{equation} 
where $\varepsilon $ is defined as
\begin{equation} \label{ZEqnNum449028} 
\varepsilon =C_{1} \omega \left(\pi _{S} \left(\gamma \left(\pi _{S} \right)=0\right)\right),   
\end{equation} 
therefore the stability condition in \eqref{ZEqnNum415455} is satisfied via
\begin{equation} \label{47)} 
\mathop{\lim }\limits_{\pi _{S} \to \infty } \Pr \left(\left|{\rm {\mathcal S}}_{I}^{\left(\pi _{S} \right)} \left(R_{j} \right)\right|>B\right)<C_{1} \omega \left(\pi _{S} \left(\gamma \left(\pi _{S} \right)=0\right)\right),   
\end{equation} 
that concludes the proof.
\end{proof}

As a corollary, \eqref{ZEqnNum985408} is satisfied for the $\zeta \left(\pi _{S} \right)$ entanglement swapping method with any $\gamma \left(\pi _{S} \right)>0$ non-zero noise, at a given $\pi _{S} $.

\subsection{Complete and Perfect Swapping Sets}

Lemma 1 extends the results for entanglement swapping at complete and perfect swapping sets.
\begin{lemma}(Noise-scaled stability at perfect and complete swapping sets) An $\zeta ^{*} \left(\pi _{S} \right)$ optimal entanglement swapping at $\gamma \left(\pi _{S} \right)=0$ is strongly stable for any complete ${\rm {\mathcal S}}^{*} \left(R_{j} \right)$ and perfect $\hat{{\rm {\mathcal S}}}\left(R_{j} \right)$ entanglement swapping set. 
\end{lemma}
\begin{proof}
Let $\omega ^{*} \left(\gamma \left(\pi _{S} \right)=0\right)$ be the weight coefficient of $\zeta ^{*} \left(\pi _{S} \right)$ at $\gamma \left(\pi _{S} \right)=0$ at a given ${\rm {\mathcal S}}^{*} \left(R_{j} \right)$, as given in \eqref{ZEqnNum770876}, with  $\zeta ^{*} \left(\pi _{S} \right)$ as
\begin{equation} \label{ZEqnNum173893} 
\zeta ^{*} \left(\pi _{S} \right)=\arg \mathop{\max }\limits_{\zeta ^{*} \left(\pi _{S} \right)\in {\rm {\mathcal S}}\left(\zeta \left(\pi _{S} \right)\right)} \left\langle \zeta ^{*} \left(\pi _{S} \right),Z_{R_{j} } \left(\pi _{S} \right)\right\rangle ,  
\end{equation} 
where ${\rm {\mathcal S}}\left(\zeta \left(\pi _{S} \right)\right)$ is the set of all possible $N!$ entanglement swapping operations at a given $\pi _{S} $, and at $N$ outgoing density matrices, $\left|{\rm {\mathcal S}}_{O} \left(R_{j} \right)\right|=N$. For a $\hat{{\rm {\mathcal S}}}\left(R_{j} \right)$ perfect entanglement swapping set, ${\rm {\mathcal S}}\left(\zeta \left(\pi _{S} \right)\right)$ is the set of all possible $N!$ entanglement swapping operations, since $\left|{\rm {\mathcal S}}_{I} \left(R_{j} \right)\right|=\left|{\rm {\mathcal S}}_{O} \left(R_{j} \right)\right|=N$ \cite{refL1,refL3}.)

Then, let $\Delta _{{\rm {\mathcal L}}} $ be a difference of the Lyapunov functions ${\rm {\mathcal L}}\left(Z_{R_{j} } \left(\pi _{S} \right)\right)$ and ${\rm {\mathcal L}}\left(Z_{R_{j} } \left(\pi '_{S} \right)\right)$, 
\begin{equation} \label{49)} 
{\rm {\mathcal L}}\left(Z_{R_{j} } \left(\pi _{S} \right)\right)=\sum _{i,k}\left(Z_{R_{j} }^{\left(\pi _{S} \right)} \left(\left(R_{i} ,\sigma _{k} \right)\right)\right)^{2}   
\end{equation} 
and  
\begin{equation} \label{50)} 
{\rm {\mathcal L}}\left(Z_{R_{j} } \left(\pi '_{S} \right)\right)=\sum _{i,k}\left(Z_{R_{j} }^{\left(\pi '_{S} \right)} \left(\left(R_{i} ,\sigma _{k} \right)\right)\right)^{2}  , 
\end{equation} 
where $\pi '_{S} $ is a next entanglement swapping period; as 
\begin{equation} \label{51)} 
\Delta _{{\rm {\mathcal L}}} ={\rm {\mathcal L}}\left(Z_{R_{j} } \left(\pi '_{S} \right)\right)-{\rm {\mathcal L}}\left(Z_{R_{j} } \left(\pi _{S} \right)\right).  
\end{equation} 
Then, for any complete swapping set ${\rm {\mathcal S}}^{*} \left(R_{j} \right)$, from \eqref{ZEqnNum588923}, ${\mathbb{E}}\left(\left. \Delta _{{\rm {\mathcal L}}} \right|Z_{R_{j} } \left(\pi _{S} \right)\right)$ at $\zeta ^{*} \left(\pi _{S} \right)$ is as
\begin{equation} \label{ZEqnNum805523} 
{\mathbb{E}}\left(\left. \Delta _{{\rm {\mathcal L}}} \right|Z_{R_{j} } \left(\pi _{S} \right)\right)\le -2C_{1} \omega ^{*} \left(\gamma \left(\pi _{S} \right)\right)+2N^{2} ,  
\end{equation} 
where $C_{1} $ is set as in \eqref{ZEqnNum400853}, and by some fundamentals of queueing theory \cite{refL1,refL2,refL3,refD1}, the condition in \eqref{ZEqnNum908737} can be rewritten as 
\begin{equation} \label{ZEqnNum748637} 
{\mathbb{E}}\left(\left. \Delta _{{\rm {\mathcal L}}} \right|Z_{R_{j} }^{*} \left(\pi _{S} \right)\right)\le -\varepsilon \left|Z_{R_{j} }^{*} \left(\pi _{S} \right)\right|,  
\end{equation} 
where $\varepsilon $ as given in \eqref{ZEqnNum449028}, while $\left|Z_{R_{j} }^{*} \left(\pi _{S} \right)\right|$ is the cardinality of the coincidence sets of ${\rm {\mathcal S}}^{*} \left(R_{j} \right)$ at a given $\pi _{S} $, as
\begin{equation} \label{54)} 
\left|Z_{R_{j} }^{*} \left(\pi _{S} \right)\right|=\sum _{i,k}Z_{R_{j} }^{\left(\pi _{S} \right)} \left(\left(R_{i} ,\sigma _{k} \right)\right) =\left|{\rm {\mathcal S}}_{I}^{\left(\pi _{S} \right)} \left(R_{j} \right)\right|.  
\end{equation} 
Thus, \eqref{ZEqnNum748637} can be rewritten as
\begin{equation} \label{ZEqnNum511923} 
{\mathbb{E}}\left(\left. \Delta _{{\rm {\mathcal L}}} \right|Z_{R_{j} }^{*} \left(\pi _{S} \right)\right)\le -\varepsilon \left|{\rm {\mathcal S}}_{I}^{\left(\pi _{S} \right)} \left(R_{j} \right)\right|.  
\end{equation} 
By similar assumptions, for any $\hat{{\rm {\mathcal S}}}\left(R_{j} \right)$ perfect entanglement swapping with cardinality $\left|\hat{Z}_{R_{j} } \left(\pi _{S} \right)\right|$ of the coincidence sets of $\hat{{\rm {\mathcal S}}}\left(R_{j} \right)$ at a given $\pi _{S} $, the condition in \eqref{ZEqnNum908737} can be rewritten as 
\begin{equation} \label{ZEqnNum182627} 
{\mathbb{E}}\left(\left. \Delta _{{\rm {\mathcal L}}} \right|\hat{Z}_{R_{j} } \left(\pi _{S} \right)\right)\le -\varepsilon \left|\hat{Z}_{R_{j} } \left(\pi _{S} \right)\right|.  
\end{equation} 
Thus, from \eqref{ZEqnNum511923} and \eqref{ZEqnNum182627}, it follows that $\zeta ^{*} \left(\pi _{S} \right)$ \eqref{ZEqnNum173893} is strongly stable for any complete and perfect entanglement swapping set, which concludes the proof.
\end{proof}
\section{Noise-Scaled Entanglement Rate Maximization}
\label{sec4}
This section proposes the entanglement rate maximization procedure for the different entanglement swapping sets.

Since the entanglement swapping is stable for both complete and non-complete entanglement swapping, this allows us to derive further results for the noise-scaled entanglement rate. The proposed derivations utilize the fundamentals of queueing theory (Note: in queueing theory, Little's law defines a connection between the $L$ average queue length and the $W$ average delay as $L=\lambda W$, where $\lambda $ is the arrival rate. The stability property is a required preliminary condition for the relation.). The derivations of the maximized noise-scaled entanglement rate assume that an incoming density matrix $\rho $ chooses a particular output density matrix $\sigma $ for the entanglement swapping with probability $\Pr \left(\rho ,\sigma \right)=x\ge 0$. 

\subsection{Preliminaries}
Let $\left|Z_{R_{j} } \left(\pi _{S} \right)\right|$ be the cardinality of the coincidence sets at a given $\pi _{S} $, as
\begin{equation} \label{ZEqnNum450845} 
\left|Z_{R_{j} } \left(\pi _{S} \right)\right|=\sum _{i,k}Z_{R_{j} }^{\left(\pi _{S} \right)} \left(\left(R_{i} ,\sigma _{k} \right)\right) ,  
\end{equation} 
and let $\left|B_{R_{i} } \left(\pi _{S} \right)\right|$ be the total number of incoming density matrices in $R_{j} $ per a given $\pi _{S} $, as
\begin{equation} \label{ZEqnNum505402} 
\left|B_{R_{j} } \left(\pi _{S} \right)\right|=\sum _{i,k}\left|B\left(R_{i} \left(\pi _{S} \right),\sigma _{k} \right)\right| ,  
\end{equation} 
where $\left|B\left(R_{i} \left(\pi _{S} \right),\sigma _{k} \right)\right|$ refers to the number of density matrices arrive from $R_{i} $ for swapping with $\sigma _{k} $ per $\pi _{S} $.

From \eqref{ZEqnNum450845} and \eqref{ZEqnNum505402}, let $D\left(\pi _{S} \right)$ be the delay measured in entanglement swapping periods, as
\begin{equation} \label{59)} 
D\left(\pi _{S} \right)={\tfrac{\left|Z_{R_{j} } \left(\pi _{S} \right)\right|}{\left|B_{R_{j} } \left(\pi _{S} \right)\right|}} =\sum _{i,k}D_{R_{j} }^{\left(\pi _{S} \right)} \left(\left(R_{i} ,\sigma _{k} \right)\right) ,   
\end{equation} 
where $D_{R_{j} }^{\left(\pi _{S} \right)} \left(\left(R_{i} ,\sigma _{k} \right)\right)$ is the delay for a given $R_{i} $ at a given $\pi _{S} $, as
\begin{equation} \label{60)} 
D_{R_{j} }^{\left(\pi _{S} \right)} \left(\left(R_{i} ,\sigma _{k} \right)\right)={\tfrac{Z_{R_{j} }^{\left(\pi _{S} \right)} \left(\left(R_{i} ,\sigma _{k} \right)\right)}{\left|B\left(R_{i} \left(\pi _{S} \right),\sigma _{k} \right)\right|}} .   
\end{equation} 
At delay $D\left(\pi _{S} \right)\ge 0$, the $B'_{R_{j} } \left(\pi _{S} \right)$ number of swapped density matrices per $\pi _{S} $ is
\begin{equation} \label{ZEqnNum517324}
\begin{split}
   \left| {{{{B}'_{{{R}_{j}}}}}}\left( {{\pi }_{S}} \right) \right|&=\left( 1-\tfrac{L}{N} \right)\tfrac{{{\pi }_{S}}}{{{{\tilde{\pi }}}_{S}}}\left| {{B}_{{{R}_{j}}}}\left( {{\pi }_{S}} \right) \right| \\ 
 & =\left( 1-\tfrac{L}{N} \right)\tfrac{x{{t}_{C}}}{x{{t}_{C}}+D\left( {{\pi }_{S}} \right)x{{t}_{C}}}\left| {{B}_{{{R}_{j}}}}\left( {{\pi }_{S}} \right) \right| \\ 
 & =\left( 1-\tfrac{L}{N} \right)\tfrac{1}{1+D\left( {{\pi }_{S}} \right)}\left| {{B}_{{{R}_{j}}}}\left( {{\pi }_{S}} \right) \right| \\ 
 & =\left( 1-\tfrac{L}{N} \right)\tfrac{{{\left| {{B}_{{{R}_{j}}}}\left( {{\pi }_{S}} \right) \right|}^{2}}}{\left| {{B}_{{{R}_{j}}}}\left( {{\pi }_{S}} \right) \right|+\left( D\left( {{\pi }_{S}} \right)\left| {{B}_{{{R}_{j}}}}\left( {{\pi }_{S}} \right) \right| \right)},  
\end{split}
\end{equation}
where $0<L\le N$ is the number of lost density matrices in ${\rm {\mathcal S}}_{O} \left(R_{j} \right)$ of $R_{j} $ per $\pi _{S} $ at a non-zero noise $\gamma \left(\pi _{S} \right)>0$, $L=0$ if $\gamma \left(\pi _{S} \right)=0$, and $\tilde{\pi }_{S} $ is the extended period, defined as
\begin{equation} \label{62)} 
\tilde{\pi }_{S} =\pi _{S} +D\left(\pi _{S} \right),  
\end{equation} 
with ${\pi _{S} \mathord{\left/ {\vphantom {\pi _{S}  \left(\tilde{\pi }_{S} \right)}} \right. \kern-\nulldelimiterspace} \left(\tilde{\pi }_{S} \right)} \le 1$; thus, \eqref{ZEqnNum517324} identifies the $B'_{R_{j} } \left(\pi _{S} \right)$ outgoing entanglement rate per $\pi _{S} $ for a particular entanglement swapping set. 

\subsection{Non-Complete Swapping Sets }
\begin{theorem}
(Entanglement rate decrement at non-complete swapping sets). For a non-complete entanglement swapping set ${\rm {\mathcal S}}\left(R_{j} \right)$, $\gamma >0$, the $B'_{R_{j} } \left(\pi _{S} \right)$ outgoing entanglement rate is $\left|B'_{R_{j} } \left(\pi _{S} \right)\right|=\left(1-{\tfrac{L}{N}} \right){\tfrac{1}{1+D\left(\pi _{S} \right)}} \left(\left|B_{R_{j} } \left(\pi _{S} \right)\right|\right)$ per $\pi _{S} $, where $\left|B_{R_{j} } \left(\pi _{S} \right)\right|$ is the total incoming entanglement throughput at a given $\pi _{S} $, and $D\left(\pi _{S} \right)\le {\left|Z_{R_{j} } \left(\pi _{S} \right)\right|\mathord{\left/ {\vphantom {\left|Z_{R_{j} } \left(\pi _{S} \right)\right| \left|B_{R_{j} } \left(\pi _{S} \right)\right|}} \right. \kern-\nulldelimiterspace} \left|B_{R_{j} } \left(\pi _{S} \right)\right|} $, where $\mathop{\lim }\limits_{\pi _{S} \to \infty } {\mathbb{E}}\left[\left|Z_{R_{j} } \left(\pi _{S} \right)\right|\right]\le {\tfrac{\left(N-L\right)\beta }{C_{1} }} +\xi \left(\gamma \right)$, where $\beta =\sum _{i,k}\left(\left|\bar{B}\left(R_{i} \left(\pi _{S} \right),\sigma _{k} \right)\right|-\left|\bar{B}\left(R_{i} \left(\pi _{S} \right),\sigma _{k} \right)\right|^{2} \right) $ and $\xi \left(\gamma \right)={\tfrac{\left(N-L\right)}{2C_{1} }} f\left(\gamma \left(\pi _{S} \right)\right)$, where $C_{1} >0$ is a constant.
\end{theorem}
\begin{proof}
The entanglement rate decrement for a non-complete entanglement swapping set is as follows. 

After some calculations, the result in \eqref{ZEqnNum588923} can be rewritten as
\begin{equation} \label{ZEqnNum461144}
\begin{split}
  & \mathbb{E}\left( \left. {{\Delta }_{\mathcal{L}}} \right|{{Z}_{{{R}_{j}}}}\left( {{\pi }_{S}} \right) \right) \\ 
 & \le -2\tfrac{{{C}_{1}}}{N-L}\left| {{Z}_{{{R}_{j}}}}\left( {{\pi }_{S}} \right) \right|+2f\left( \gamma \left( {{\pi }_{S}} \right) \right)+2\beta ,  
\end{split}
\end{equation}
where
\begin{equation} \label{64)} 
C_{1} =1-\mathop{\max }\limits_{i} \left(\sum _{k}\left|\bar{B}\left(R_{i} \left(\pi _{S} \right),\sigma _{k} \right)\right| \right), 
\end{equation} 
where $\left|\bar{B}\left(R_{i} \left(\pi _{S} \right),\sigma _{k} \right)\right|$ refers to the normalized number of density matrices arrive from $R_{i} $ for swapping with $\sigma _{k} $ at $\pi _{S} $ as
\begin{equation} \label{65)} 
\left|\bar{B}\left(R_{i} \left(\pi _{S} \right),\sigma _{k} \right)\right|={\tfrac{\left|B\left(R_{i} \left(\pi _{S} \right),\sigma _{k} \right)\right|}{\sum _{i}\left|B\left(R_{i} \left(\pi _{S} \right),\sigma _{k} \right)\right| }} ,  
\end{equation} 
while $\beta $ is defined as \cite{refL1,refD1}
\begin{equation} \label{66)} 
\beta =\sum _{i,k}\left(\left|\bar{B}\left(R_{i} \left(\pi _{S} \right),\sigma _{k} \right)\right|-\left|\bar{B}\left(R_{i} \left(\pi _{S} \right),\sigma _{k} \right)\right|^{2} \right) .   
\end{equation} 
From \eqref{ZEqnNum461144}, ${\mathbb{E}}\left({\rm {\mathcal L}}\left(Z_{R_{j} } \left(\pi '_{S} \right)\right)\right)$ can be evaluated as
\begin{equation}\label{67)}
\begin{split}
   \mathbb{E}\left( \mathcal{L}\left( {{Z}_{{{R}_{j}}}}\left( {{{{\pi }'_{S}}}} \right) \right) \right)&=\mathbb{E}\left( \mathcal{L}\left( {{Z}_{{{R}_{j}}}}\left( {{{{\pi }'_{S}}}} \right) \right)-\mathcal{L}\left( {{Z}_{{{R}_{j}}}}\left( {{\pi }_{S}} \right) \right)+\mathcal{L}\left( {{Z}_{{{R}_{j}}}}\left( {{\pi }_{S}} \right) \right) \right) \\ 
 & =\mathbb{E}\left( \mathbb{E}\left( \mathcal{L}\left( {{Z}_{{{R}_{j}}}}\left( {{{{\pi }'_{S}}}} \right) \right)-\mathcal{L}\left. \left( {{Z}_{{{R}_{j}}}}\left( {{\pi }_{S}} \right) \right) \right|{{Z}_{{{R}_{j}}}}\left( {{\pi }_{S}} \right) \right) \right)+\mathbb{E}\left( \mathcal{L}\left( {{Z}_{{{R}_{j}}}}\left( {{\pi }_{S}} \right) \right) \right) \\ 
 & \le -2\tfrac{{{C}_{1}}}{N-L}\mathbb{E}\left( \left| {{Z}_{{{R}_{j}}}}\left( {{\pi }_{S}} \right) \right| \right)+2f\left( \gamma \left( {{\pi }_{S}} \right) \right)+2\beta +\mathbb{E}\left( \mathcal{L}\left( {{Z}_{{{R}_{j}}}}\left( {{\pi }_{S}} \right) \right) \right),  
\end{split}
\end{equation}
thus, after $P$ entanglement swapping periods $\pi _{S} =0,\ldots ,P-1$, ${\mathbb{E}}\left({\rm {\mathcal L}}\left(Z_{R_{j} } \left(P\pi _{S} \right)\right)\right)$ is yielded as
\begin{equation} \label{68)} 
{\mathbb{E}}\left({\rm {\mathcal L}}\left(Z_{R_{j} } \left(\pi _{S} \left(P\right)\right)\right)\right)\le P\left(2f\left(\gamma \left(\pi _{S} \right)\right)+2\beta \right)+{\mathbb{E}}\left({\rm {\mathcal L}}\left(Z_{R_{j} } \left(\pi _{S} \left(0\right)\right)\right)\right)-2{\tfrac{C_{1} }{N-L}} \sum _{\pi _{S} =0}^{P-1}{\mathbb{E}}\left(\left|Z_{R_{j} } \left(\pi _{S} \right)\right|\right) ,   
\end{equation} 
where $\pi _{S} \left(P\right)$ is the $P$-th entanglement swapping, while ${\mathbb{E}}\left({\rm {\mathcal L}}\left(Z_{R_{j} } \left(\pi _{S} \left(0\right)\right)\right)\right)$ identifies the initial system state, thus
\begin{equation} \label{69)} 
{\mathbb{E}}\left({\rm {\mathcal L}}\left(Z_{R_{j} } \left(\pi _{S} \left(0\right)\right)\right)\right)=0,  
\end{equation} 
by theory \cite{refD1}. 

Therefore, after $P$ entanglement swapping periods, the expected value of $\left|Z_{R_{j} } \left(\pi _{S} \right)\right|$ can be evaluated as
\begin{equation} \label{ZEqnNum472017} 
{\tfrac{1}{P}} \sum _{\pi _{S} =0}^{P-1}{\mathbb{E}}\left(\left|Z_{R_{j} } \left(\pi _{S} \right)\right|\right) \le {\tfrac{N-L}{2C_{1} }} \left(2f\left(\gamma \left(\pi _{S} \right)\right)+2\beta \right)-{\tfrac{1}{P}} {\mathbb{E}}\left({\rm {\mathcal L}}\left(Z_{R_{j} } \left(\pi _{S} \left(P\right)\right)\right)\right),  
\end{equation} 
where ${\rm {\mathcal L}}\left(Z_{R_{j} } \left(\pi _{S} \left(P\right)\right)\right)\ge 0$, thus, assuming that the arrival of the density matrices can be modeled as an i.i.d. arrival process, the result in \eqref{ZEqnNum472017} can be rewritten as
\begin{equation} \label{ZEqnNum474049} 
\mathop{\lim }\limits_{\pi _{S} \to \infty } {\tfrac{1}{P}} \sum _{\pi _{S} =0}^{P-1}{\mathbb{E}}\left(\left|Z_{R_{j} } \left(\pi _{S} \right)\right|\right)=\mathop{\lim }\limits_{\pi _{S} \to \infty } {\mathbb{E}}\left(\left|Z_{R_{j} } \left(\pi _{S} \right)\right|\right) \le {\tfrac{N-L}{2C_{1} }} \left(2f\left(\gamma \left(\pi _{S} \right)\right)+2\beta \right).  
\end{equation} 
Then, since for any noise $\gamma \left(\pi _{S} \right)$ at $\pi _{S} $, the relation
\begin{equation} \label{72)} 
\gamma \left(\pi _{S} \right)\le \left|Z_{R_{j} } \left(\pi _{S} \right)\right| 
\end{equation} 
holds, thus for any entanglement swapping period $\pi _{S} $, from the sub-linear property of $f\left(\cdot \right)$, the relation 
\begin{equation} \label{73)} 
f\left(\gamma \left(\pi _{S} \right)\right)\le f\left(\left|Z_{R_{j} } \left(\pi _{S} \right)\right|\right) 
\end{equation} 
follows for $f\left(\gamma \left(\pi _{S} \right)\right)$.

Therefore, $\omega \left(\gamma \left(\pi _{S} \right)\right)$ from \eqref{ZEqnNum577058} can be rewritten as
\begin{equation} \label{ZEqnNum446340} 
\omega \left(\gamma \left(\pi _{S} \right)\right)\ge \omega ^{*} \left(\gamma \left(\pi _{S} \right)=0\right)-f\left(\left|Z_{R_{j} } \left(\pi _{S} \right)\right|\right),  
\end{equation} 
such that for $f\left(\left|Z_{R_{j} } \left(\pi _{S} \right)\right|\right)$, the relation
\begin{equation} \label{75)} 
\mathop{\lim }\limits_{\pi _{S} \to \infty } {\tfrac{1}{P}} \sum _{\pi _{S} =0}^{P-1}f\left(\left|Z_{R_{j} } \left(\pi _{S} \right)\right|\right)=\mathop{\lim }\limits_{\pi _{S} \to \infty } {\mathbb{E}}\left(f\left(\left|Z_{R_{j} } \left(\pi _{S} \left(P\right)\right)\right|\right)\right) ,  
\end{equation} 
holds, which allows us to rewrite \eqref{ZEqnNum474049} in the following manner:
\begin{equation} \label{ZEqnNum512518} 
\mathop{\lim }\limits_{\pi _{S} \to \infty } {\mathbb{E}}\left(\left|Z_{R_{j} } \left(\pi _{S} \right)\right|\right)\le {\tfrac{\left(N-L\right)\beta }{C_{1} }} +{\tfrac{N-L}{2C_{1} }} \mathop{\lim }\limits_{\pi _{S} \to \infty } {\mathbb{E}}\left(f\left(\left|Z_{R_{j} } \left(\pi _{S} \right)\right|\right)\right),  
\end{equation} 
where
\begin{equation} \label{ZEqnNum482205} 
\mathop{\lim }\limits_{\pi _{S} \to \infty } {\mathbb{E}}\left(f\left(\left|Z_{R_{j} } \left(\pi _{S} \right)\right|\right)\right)=f\left(\gamma \left(\pi _{S} \right)\right),  
\end{equation} 
thus from \eqref{ZEqnNum482205}, \eqref{ZEqnNum512518} is as 
\begin{equation} \label{ZEqnNum557989} 
\mathop{\lim }\limits_{\pi _{S} \to \infty } {\mathbb{E}}\left(\left|Z_{R_{j} } \left(\pi _{S} \right)\right|\right)\le {\tfrac{\left(N-L\right)\beta }{C_{1} }} +{\tfrac{N-L}{2C_{1} }} f\left(\gamma \left(\pi _{S} \right)\right)={\tfrac{\left(N-L\right)\beta }{C_{1} }} +\xi \left(\gamma \right),  
\end{equation} 
where 
\begin{equation} \label{79)} 
\xi \left(\gamma \right)={\tfrac{\left(N-L\right)}{2C_{1} }} f\left(\gamma \left(\pi _{S} \right)\right).  
\end{equation} 
Therefore, the $D\left(\pi _{S} \right)$ delay for any non-complete swapping set is as
\begin{equation} \label{ZEqnNum102271} 
D\left(\pi _{S} \right)\le {\tfrac{\left|Z_{R_{j} } \left(\pi _{S} \right)\right|}{\left|B_{R_{j} } \left(\pi _{S} \right)\right|}} ={\tfrac{1}{\left|B_{R_{j} } \left(\pi _{S} \right)\right|}} \left({\tfrac{\left(N-L\right)\beta }{C_{1} }} +{\tfrac{N-L}{2C_{1} }} f\left(\gamma \left(\pi _{S} \right)\right)\right).   
\end{equation} 
As a corollary, the $B'_{R_{j} } \left(\pi _{S} \right)$ outgoing entanglement rate per $\pi _{S} $ at \eqref{ZEqnNum102271}, is as
\begin{equation}\label{ZEqnNum901920}
\begin{split}
   \left| {{{{B}'_{{{R}_{j}}}}}}\left( {{\pi }_{S}} \right) \right|&=\left( 1-\tfrac{L}{N} \right)\tfrac{1}{1+D\left( {{\pi }_{S}} \right)}\left( \left| {{B}_{{{R}_{j}}}}\left( {{\pi }_{S}} \right) \right| \right) \\ 
 & =\left( 1-\tfrac{L}{N} \right)\tfrac{{{\left| {{B}_{{{R}_{j}}}}\left( {{\pi }_{S}} \right) \right|}^{2}}}{\left| {{B}_{{{R}_{j}}}}\left( {{\pi }_{S}} \right) \right|+\left( D\left( {{\pi }_{S}} \right)\left| {{B}_{{{R}_{j}}}}\left( {{\pi }_{S}} \right) \right| \right)},  
\end{split}
\end{equation} 
that concludes the proof.
\end{proof}
\subsection{Complete Swapping Sets }
\begin{theorem}
(Entanglement rate decrement at complete swapping sets). For a complete entanglement swapping set ${\rm {\mathcal S}}^{*} \left(R_{j} \right)$, $\gamma =0$, the $B'_{R_{j} } \left(\pi _{S} \right)$ outgoing entanglement rate is $\left|B'_{R_{j} } \left(\pi _{S} \right)\right|={\tfrac{1}{1+D^{*} \left(\pi _{S} \right)}} \left|B_{R_{j} } \left(\pi _{S} \right)\right|$ per $\pi _{S} $, where $D^{*} \left(\pi _{S} \right)={\left|Z_{R_{j} }^{*} \left(\pi _{S} \right)\right|\mathord{\left/ {\vphantom {\left|Z_{R_{j} }^{*} \left(\pi _{S} \right)\right| \left|B_{R_{j} } \left(\pi _{S} \right)\right|}} \right. \kern-\nulldelimiterspace} \left|B_{R_{j} } \left(\pi _{S} \right)\right|} <D\left(\pi _{S} \right)$, with $\mathop{\lim }\limits_{\pi _{S} \to \infty } {\mathbb{E}}\left[\left|Z_{R_{j} }^{*} \left(\pi _{S} \right)\right|\right]\le {\tfrac{N\beta }{C_{1} }} $.
\end{theorem}
\begin{proof}
Since for any complete swapping set, 
\begin{equation} \label{82)} 
f\left(\gamma \left(\pi _{S} \right)\right)=0,   
\end{equation} 
it follows that
\begin{equation} \label{83)} 
\xi \left(\gamma \right)=0,   
\end{equation} 
therefore \eqref{ZEqnNum557989} can be rewritten for a complete swapping set as
\begin{equation} \label{84)} 
\mathop{\lim }\limits_{\pi _{S} \to \infty } {\mathbb{E}}\left(\left|Z_{R_{j} }^{*} \left(\pi _{S} \right)\right|\right)\le {\tfrac{N\beta }{C_{1} }} .   
\end{equation} 
Therefore, the $D^{*} \left(\pi _{S} \right)$ decrement for any complete swapping set is as
\begin{equation} \label{ZEqnNum531673} 
D^{*} \left(\pi _{S} \right)\le {\tfrac{\left|Z_{R_{j} }^{*} \left(\pi _{S} \right)\right|}{\left|B_{R_{j} } \left(\pi _{S} \right)\right|}} ={\tfrac{N\beta }{C_{1} \left|B_{R_{j} } \left(\pi _{S} \right)\right|}} ,   
\end{equation} 
with relation $D^{*} \left(\pi _{S} \right)<D\left(\pi _{S} \right)$, where $D\left(\pi _{S} \right)$ is as in \eqref{ZEqnNum102271}. 

As a corollary, the $B'_{R_{j} } \left(\pi _{S} \right)$ entanglement rate at \eqref{ZEqnNum531673}, is as
\begin{equation}\label{86)}
\begin{split}
   \left| {{{{B}'_{{{R}_{j}}}}}}\left( {{\pi }_{S}} \right) \right|&=\tfrac{1}{1+{{D}^{*}}\left( {{\pi }_{S}} \right)}\left| {{B}_{{{R}_{j}}}}\left( {{\pi }_{S}} \right) \right| \\ 
 & =\tfrac{{{\left| {{B}_{{{R}_{j}}}}\left( {{\pi }_{S}} \right) \right|}^{2}}}{\left| {{B}_{{{R}_{j}}}}\left( {{\pi }_{S}} \right) \right|+\left( {{D}^{*}}\left( {{\pi }_{S}} \right)\left| {{B}_{{{R}_{j}}}}\left( {{\pi }_{S}} \right) \right| \right)},  
\end{split}
\end{equation}
which concludes the proof.
\end{proof}

\subsubsection{Perfect Swapping Sets}
\begin{lemma}
(Entanglement rate decrement at perfect swapping sets). For a perfect entanglement swapping set $\hat{{\rm {\mathcal S}}}\left(R_{j} \right)$, $\gamma =0$, the $B'_{R_{j} } \left(\pi _{S} \right)$ outgoing entanglement rate is $\left|B'_{R_{j} } \left(\pi _{S} \right)\right|={\tfrac{1}{1+\hat{D}\left(\pi _{S} \right)}} \left|B_{R_{j} } \left(\pi _{S} \right)\right|$ per $\pi _{S} $, where $\hat{D}\left(\pi _{S} \right)={\left|\hat{Z}_{R_{j} } \left(\pi _{S} \right)\right|\mathord{\left/ {\vphantom {\left|\hat{Z}_{R_{j} } \left(\pi _{S} \right)\right| \left|B_{R_{j} } \left(\pi _{S} \right)\right|}} \right. \kern-\nulldelimiterspace} \left|B_{R_{j} } \left(\pi _{S} \right)\right|} $, with ${\mathbb{E}}\left[\left|\hat{Z}_{R_{j} } \left(\pi _{S} \right)\right|\right]=N$.
\end{lemma}
\begin{proof}
The proof trivially follows from the fact, that for a $\hat{{\rm {\mathcal S}}}\left(R_{j} \right)$ perfect entanglement swapping set, the $Z_{R_{j} }^{\left(\pi _{S} \right)} \left(\left(R_{i} ,\sigma _{k} \right)\right)$ cardinality of all $N$ coincidence sets ${\rm {\mathcal S}}_{R_{j} }^{\left(\pi _{S} \right)} \left(\left(R_{i} ,\sigma _{k} \right)\right)$, $i=1,\ldots ,N$, $k=1,\ldots ,N$, is $Z_{R_{j} }^{\left(\pi _{S} \right)} \left(\left(R_{i} ,\sigma _{k} \right)\right)=1$, thus 
\begin{equation} \label{87)} 
{\mathbb{E}}\left[\left|\hat{Z}_{R_{j} } \left(\pi _{S} \right)\right|\right]=N.   
\end{equation} 
As follows,
\begin{equation} \label{ZEqnNum669177} 
\hat{D}\left(\pi _{S} \right)\le {\tfrac{\left|\hat{Z}_{R_{j} } \left(\pi _{S} \right)\right|}{\left|B_{R_{j} } \left(\pi _{S} \right)\right|}} ={\tfrac{N}{\left|B_{R_{j} } \left(\pi _{S} \right)\right|}} ,  
\end{equation} 
and the $B'_{R_{j} } \left(\pi _{S} \right)$ entanglement rate at \eqref{ZEqnNum669177}, is as
\begin{equation}\label{89)}
\begin{split}
   \left| {{{{B}'_{{{R}_{j}}}}}}\left( {{\pi }_{S}} \right) \right|&=\tfrac{1}{1+\hat{D}\left( {{\pi }_{S}} \right)}\left| {{B}_{{{R}_{j}}}}\left( {{\pi }_{S}} \right) \right| \\ 
 & =\tfrac{{{\left| {{B}_{{{R}_{j}}}}\left( {{\pi }_{S}} \right) \right|}^{2}}}{\left| {{B}_{{{R}_{j}}}}\left( {{\pi }_{S}} \right) \right|+\left( \hat{D}\left( {{\pi }_{S}} \right)\left| {{B}_{{{R}_{j}}}}\left( {{\pi }_{S}} \right) \right| \right)}.  
\end{split}
\end{equation}
The proof is concluded here.
\end{proof}

\subsection{Sub-Linear Function at a Swapping Period}
\begin{theorem}
For any $\gamma \left(\pi _{S} \right)>0$, at a given $\pi _{S}^{*} =\left(1+h\right)\pi _{S} $ entanglement swapping period, $h>0$, $f\left(\gamma \left(\pi _{S} \right)\right)$ can be evaluated as $f\left(\gamma \left(\pi _{S} \right)\right)=2\pi _{S}^{*} N$.
\end{theorem}
\begin{proof}
Let $\zeta ^{*} \left(\pi _{S} \right)$ be the optimal entanglement swapping at $\gamma \left(\pi _{S} \right)=0$with a maximized weight coefficient $\omega ^{*} \left(\pi _{S} \right)$, and let $\chi \left(\pi _{S} \right)$ be an arbitrary entanglement swapping with defined as
\begin{equation} \label{ZEqnNum702785} 
\chi \left(\pi _{S} \right)=\zeta ^{*} \left(\pi _{S} -\pi _{S}^{*} \right),  
\end{equation} 
where $\zeta ^{*} \left(\pi _{S} -\pi _{S}^{*} \right)$ an optimal entanglement swapping at an $\left(\pi _{S} -\pi _{S}^{*} \right)$-th entanglement swapping period, while $\pi _{S}^{*} $ is an entanglement swapping period, defined as
\begin{equation} \label{91)} 
\pi _{S}^{*} =\left(1+h\right)\pi _{S} =\left(1+h\right)xt_{C} .  
\end{equation} 
where $h>0$. 

Then, the $\omega \left(\pi _{S} -\pi _{S}^{*} \right)$ weight coefficient of entanglement swapping $\chi \left(\pi _{S} \right)$ \eqref{ZEqnNum702785} at an $\left(\pi _{S} -\pi _{S}^{*} \right)$-th entanglement swapping period is as
\begin{equation} \label{ZEqnNum223298}
\begin{split}
   \omega \left( {{\pi }_{S}}-\pi _{S}^{*} \right)&=\sum\limits_{i,k}{{{\chi }_{ik}}\left( {{\rho }_{A}},{{\sigma }_{k}} \right)Z_{{{R}_{j}}}^{\left( {{\pi }_{S}}-\pi _{S}^{*} \right)}\left( \left( {{R}_{i}},{{\sigma }_{k}} \right) \right)} \\ 
 & =\left\langle \chi \left( {{\pi }_{S}} \right),{{Z}_{{{R}_{j}}}}\left( {{\pi }_{S}}-\pi _{S}^{*} \right) \right\rangle ,  
\end{split}
\end{equation}
while $\omega \left(\pi _{S} \right)$ at $\pi _{S} $ is as

\begin{equation} \label{ZEqnNum179592}
\begin{split}
   \omega \left( {{\pi }_{S}} \right)&=\sum\limits_{i,k}{{{\chi }_{ik}}\left( {{\rho }_{A}},{{\sigma }_{k}} \right)Z_{{{R}_{j}}}^{\left( {{\pi }_{S}} \right)}\left( \left( {{R}_{i}},{{\sigma }_{k}} \right) \right)} \\ 
 & =\left\langle \chi \left( {{\pi }_{S}} \right),{{Z}_{{{R}_{j}}}}\left( {{\pi }_{S}} \right) \right\rangle .  
\end{split}
\end{equation}
It can be concluded, that the difference of \eqref{ZEqnNum223298} and \eqref{ZEqnNum179592} is as
\begin{equation} \label{ZEqnNum295151} 
\omega \left(\pi _{S} -\pi _{S}^{*} \right)-\omega \left(\pi _{S} \right)\le \pi _{S}^{*} N, 
\end{equation} 
thus \eqref{ZEqnNum223298} is at most $\pi _{S}^{*} N$ more than \eqref{ZEqnNum179592}, since the weight coefficient of $\chi \left(\pi _{S} \right)$ can at most decrease by $N$ every period (i.e., at a given entanglement swapping period at most $N$ density matrix pairs can be swapped by $\chi \left(\pi _{S} \right)$). 

On the other hand, let 
\begin{equation}\label{ZEqnNum671774}
\begin{split}
   {{\omega }^{*}}\left( {{\pi }_{S}} \right)&=\omega \left( {{\pi }_{S}}\left( \gamma =0 \right) \right) \\ 
 & =\sum\limits_{i,k}{\zeta _{ik}^{*}\left( {{\rho }_{A}},{{\sigma }_{k}} \right)Z_{{{R}_{j}}}^{\left( {{\pi }_{S}} \right)}\left( \left( {{R}_{i}},{{\sigma }_{k}} \right) \right)} \\ 
 & =\left\langle {{\zeta }^{*}}\left( {{\pi }_{S}} \right),{{Z}_{{{R}_{j}}}}\left( {{\pi }_{S}} \right) \right\rangle   
\end{split}
\end{equation}
be the weight coefficient of $\zeta ^{*} \left(\pi _{S} \right)$ at $\pi _{S} $. It also can be verified, that the corresponding relation for the difference of \eqref{ZEqnNum223298} and \eqref{ZEqnNum179592} is as
\begin{equation} \label{96)} 
\omega ^{*} \left(\pi _{S} \right)-\omega \left(\pi _{S} -\pi _{S}^{*} \right)\le \pi _{S}^{*} N.  
\end{equation} 
From \eqref{ZEqnNum295151} and \eqref{ZEqnNum671774}, for the $\omega \left(\pi _{S} \right)$ coefficient of $\chi \left(\pi _{S} \right)$ at $\pi _{S} $, it follows that
\begin{equation} \label{97)} 
\omega \left(\pi _{S} \right)\ge \omega ^{*} \left(\pi _{S} \right)-2\pi _{S}^{*} N,   
\end{equation} 
therefore, function $f\left(\gamma \left(\pi _{S} \right)\right)$ for any non-zero noise, $\gamma \left(\pi _{S} \right)>0$, is evaluated as (i.e., $\chi \left(\pi _{S} \right)\ne \zeta ^{*} \left(\pi _{S} \right)$)
\begin{equation} \label{98)} 
f\left(\gamma \left(\pi _{S} \right)\right)=2\pi _{S}^{*} N,   
\end{equation} 
while for any $\gamma \left(\pi _{S} \right)=0$ (i.e., $\chi \left(\pi _{S} \right)=\zeta ^{*} \left(\pi _{S} \right)$)
\begin{equation} \label{99)} 
f\left(\gamma \left(\pi _{S} \right)\right)=0,   
\end{equation} 
which concludes the proof.
\end{proof}

\section{Performance Evaluation}
\label{sec5}
In this section, a numerical performance evaluation is proposed to study the delay and the entanglement rate at the different entanglement swapping sets.
\subsection{Entanglement Swapping Period}
In \fref{fig3}(a), the values of $\pi _{S}^{*} $  as a function of $h$ and $\pi _{S} $ are depicted. $\pi _{S} \in \left[1,10\right]$, $h\in \left[0,1\right]$. In \fref{fig3}(b), the values of $f\left(\gamma \left(\pi _{S} \right)\right)=2\pi _{S}^{*} N$ are depicted as a function of $h$ and $N$. 

\begin{center}
\begin{figure*}[!htbp]
\begin{center}
\includegraphics[angle = 0,width=1\linewidth]{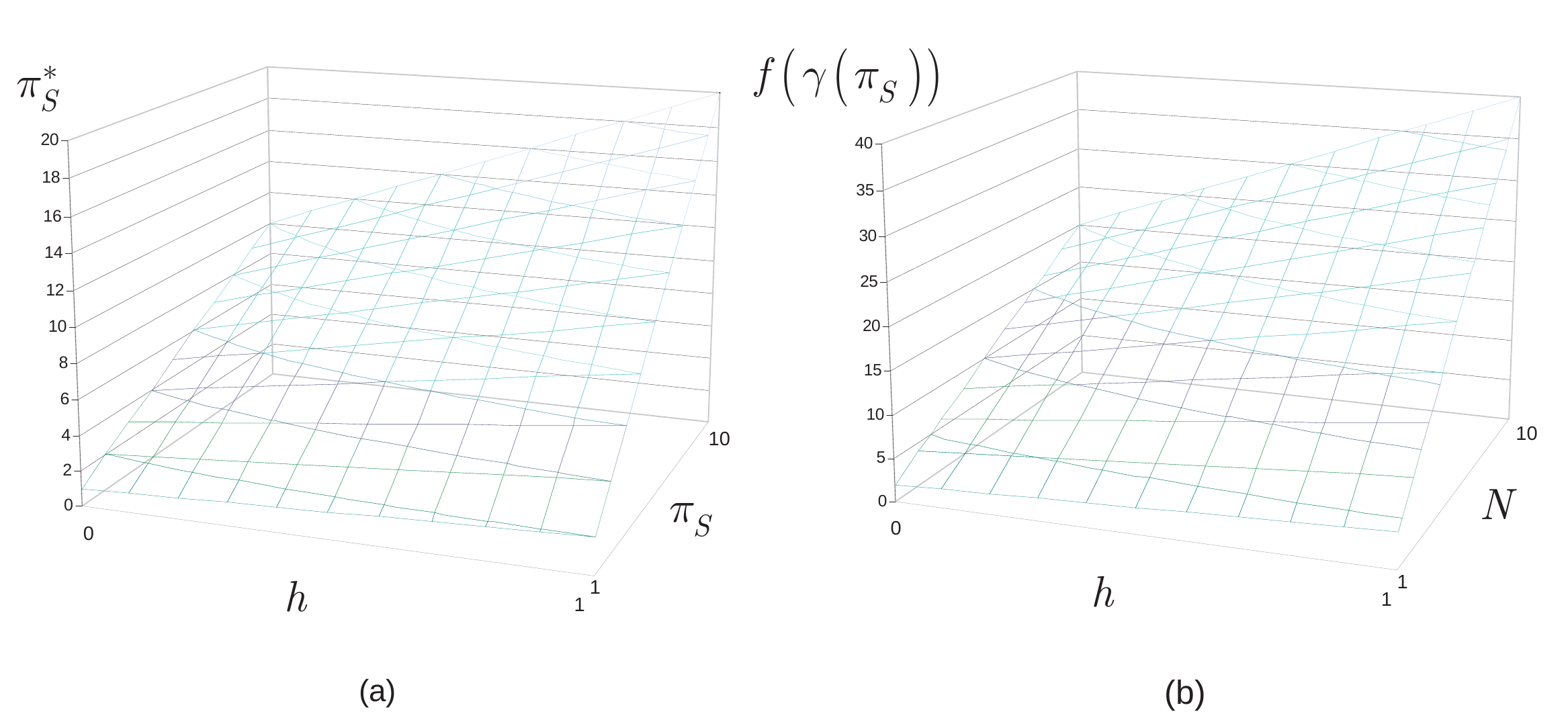}
\caption{(a) The values of $\pi _{S}^{*} $, $\pi _{S}^{*} =\left(1+h\right)\pi _{S} $, as a function of $h$, $h\in \left[0,1\right]$ and $\pi _{S} $, $\pi _{S} \in \left[1,10\right]$. (b) The values of $f\left(\gamma \left(\pi _{S} \right)\right)=2\pi _{S}^{*} N$ as a function of $h$ and $N$, $N\in \left[1,10\right]$.} 
 \label{fig3}
 \end{center}
\end{figure*}
\end{center}

\subsection{Delay and Entanglement Rate Ratio}
In \fref{fig4}(a), the values of $D\left(\pi _{S} \right)$ are depicted as a function of the incoming entanglement rate $\left|B_{R_{j} } \left(\pi _{S} \right)\right|$ for the non-complete ${\rm {\mathcal S}}\left(R_{j} \right)$, complete ${\rm {\mathcal S}}^{*} \left(R_{j} \right)$ and perfect $\hat{{\rm {\mathcal S}}}\left(R_{j} \right)$ entanglement swapping sets. The $D\left(\pi _{S} \right)$ delay values are evaluated as $D\left(\pi _{S} \right)={\left|Z_{R_{j} } \left(\pi _{S} \right)\right|\mathord{\left/ {\vphantom {\left|Z_{R_{j} } \left(\pi _{S} \right)\right| \left|B_{R_{j} } \left(\pi _{S} \right)\right|}} \right. \kern-\nulldelimiterspace} \left|B_{R_{j} } \left(\pi _{S} \right)\right|} $ via the maximized values of \eqref{ZEqnNum102271}, \eqref{ZEqnNum669177} and \eqref{ZEqnNum669177}, i.e., the delay depends on the cardinality of the coincidence set and the incoming entanglement rate (This relation also can be derived from Little's law; for details, see the fundamentals of queueing theory \cite{refL1,refL2,refL3,refD1}). 

In \fref{fig4}(b), the ratio $r$ of the $\left|B'_{R_{j} } \left(\pi _{S} \right)\right|$ outgoing and $\left|B_{R_{j} } \left(\pi _{S} \right)\right|$ incoming entanglement rates is depicted as a function of the incoming entanglement rate. For the non-complete entanglement swapping set, the loss is set as $\tilde{L}=0.2N$. 

\begin{center}
\begin{figure*}[!htbp]
\begin{center}
\includegraphics[angle = 0,width=1\linewidth]{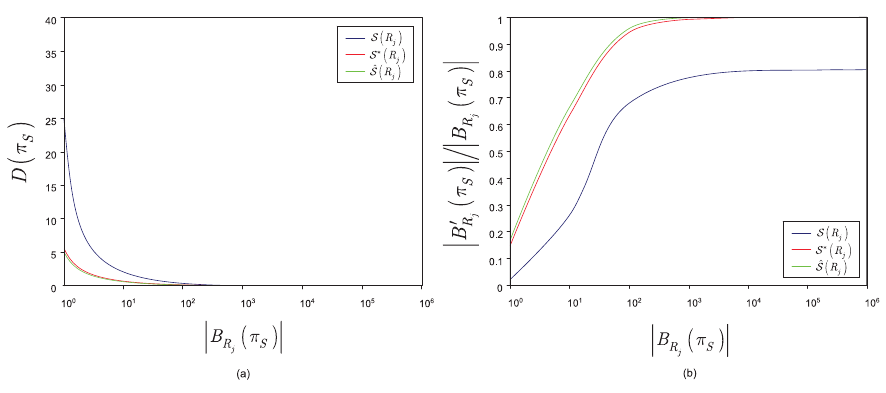}
\caption{(a) The $D\left(\pi _{S} \right)$ delay values for the different entanglement swapping sets as a function of the $\left|B_{R_{j} } \left(\pi _{S} \right)\right|$ incoming entanglement rate, $\left|B_{R_{j} } \left(\pi _{S} \right)\right|\in \left[10^{0} ,10^{8} \right]$, $N=5$, $\tilde{L}=0.2N$, $\beta =0.78$ for the complete and perfect sets, $\beta =0.64$ for a non-complete set and $C_{1} =0.7$, $\gamma \left(\pi _{S} \right)=0.2$, $h=0.2$, $\pi _{S}^{*} =1.2\pi _{S} $ and $f\left(\gamma \left(\pi _{S} \right)\right)=2\pi _{S}^{*} N=12$. (b) The ratio $r={\left| {{{{B}'_{{{R}_{j}}}}}}\left( {{\pi }_{S}} \right) \right|}/{\left| {{B}_{{{R}_{j}}}}\left( {{\pi }_{S}} \right) \right|}$ of the $\left| {{{{B}'_{{{R}_{j}}}}}}\left( {{\pi }_{S}} \right) \right|$ outgoing and $\left| {{B}_{{{R}_{j}}}}\left( {{\pi }_{S}} \right) \right|$ incoming entanglement rates for the different entanglement swapping sets as a function of the $\left| {{B}_{{{R}_{j}}}}\left( {{\pi }_{S}} \right) \right|$ incoming entanglement rate. The entanglement rate decrease for the non-complete swapping set caused by losses is $\tfrac{{\tilde{L}}}{N}\left| {{B}_{{{R}_{j}}}}\left( {{\pi }_{S}} \right) \right|$, while $\tilde{L}=0$ for the complete and perfect entanglement swapping sets.} 
 \label{fig4}
 \end{center}
\end{figure*}
\end{center}

The highest $D\left(\pi _{S} \right)$ delay values can be obtained for non-complete entanglement swapping set ${\rm {\mathcal S}}\left(R_{j} \right)$, while the lowest delays can be found for the perfect entanglement swapping set $\hat{{\rm {\mathcal S}}}\left(R_{j} \right)$. For a complete entanglement swapping set ${\rm {\mathcal S}}^{*} \left(R_{j} \right)$, the delay values are between the non-complete and perfect sets. This is because for a non-complete set, the losses due to the $\gamma \left(\pi _{S} \right)>0$ non-zero noise allow only an approximation of the delay of a complete set; thus, $D_{R_{j} } \left(\pi _{S} \right)>D_{R_{j} }^{*} \left(\pi _{S} \right)$. For a complete entanglement swapping set, while the noise is zero, $\gamma \left(\pi _{S} \right)=0$, the $\left|Z_{R_{j} }^{*} \left(\pi _{S} \right)\right|$ cardinality of the coincidence set is high, $\left|Z_{R_{j} }^{*} \left(\pi _{S} \right)\right|>\left|\hat{Z}_{R_{j} } \left(\pi _{S} \right)\right|$; thus, the $D_{R_{j} }^{*} \left(\pi _{S} \right)$ delay is higher than the $\hat{D}_{R_{j} } \left(\pi _{S} \right)$ delay of a perfect entanglement swapping set, $D_{R_{j} }^{*} \left(\pi _{S} \right)>\hat{D}_{R_{j} } \left(\pi _{S} \right)$. For a perfect set, the noise is zero, $\gamma \left(\pi _{S} \right)=0$ and the cardinality of the coincidence sets is one for all inputs; thus, $\left|Z_{R_{j} }^{*} \left(\pi _{S} \right)\right|>\left|\hat{Z}_{R_{j} } \left(\pi _{S} \right)\right|=N$. Therefore, the $\hat{D}_{R_{j} } \left(\pi _{S} \right)$ delay is minimal for a perfect set.

From the relation of $D_{R_{j} } \left(\pi _{S} \right)>D_{R_{j} }^{*} \left(\pi _{S} \right)>\hat{D}_{R_{j} } \left(\pi _{S} \right)$, the corresponding delays of the entanglement swapping sets, the relation for the decrease in the outgoing entanglement rates is straightforward, as follows. As $r\to 1$ holds for the ratio $r$ of the $\left|B'_{R_{j} } \left(\pi _{S} \right)\right|$ outgoing and $\left|B_{R_{j} } \left(\pi _{S} \right)\right|$ incoming entanglement rates, then $\left|B'_{R_{j} } \left(\pi _{S} \right)\right|\to \left|B_{R_{j} } \left(\pi _{S} \right)\right|$, i.e., no significant decrease is caused by the entanglement swapping operation.

The highest outgoing entanglement rates are obtained for a perfect entanglement swapping set, which is followed by the outgoing rates at a complete entanglement swapping set. For a non-complete set, the outgoing rate is significantly lower due to the losses caused by the non-zero noise in comparison with the perfect and complete sets.

\section{Conclusions}
\label{sec6}
The quantum repeaters determine the structure and performance attributes of the quantum Internet. Here, we defined the theory of noise-scaled stability derivation of the quantum repeaters and methods of entanglement rate maximization for the quantum Internet. The framework characterized the stability conditions of entanglement swapping in quantum repeaters and the terms of non-complete, complete and perfect entanglement swapping sets in the quantum repeaters to model the status of the quantum memory of the quantum repeaters. The defined terms are evaluated as a function of the noise level of the quantum repeaters to describe the physical procedures of the quantum repeaters. We derived the conditions for an optimal entanglement swapping at a particular noise level to maximize the entanglement throughput of the quantum repeaters. The results are applicable to the experimental quantum Internet.

\section*{Acknowledgements}
The research reported in this paper has been supported by the Hungarian Academy of Sciences (MTA Premium Postdoctoral Research Program 2019), by the National Research, Development and Innovation Fund (TUDFO/51757/2019-ITM, Thematic Excellence Program), by the National Research Development and Innovation Office of Hungary (Project No. 2017-1.2.1-NKP-2017-00001), by the Hungarian Scientific Research Fund - OTKA K-112125 and in part by the BME Artificial Intelligence FIKP grant of EMMI (Budapest University of Technology, BME FIKP-MI/SC).


\newpage
\appendix
\setcounter{table}{0}
\setcounter{figure}{0}
\setcounter{equation}{0}
\setcounter{algocf}{0}
\renewcommand{\thetable}{\Alph{section}.\arabic{table}}
\renewcommand{\thefigure}{\Alph{section}.\arabic{figure}}
\renewcommand{\theequation}{\Alph{section}.\arabic{equation}}
\renewcommand{\thealgocf}{\Alph{section}.\arabic{algocf}}

\setlength{\arrayrulewidth}{0.1mm}
\setlength{\tabcolsep}{5pt}
\renewcommand{\arraystretch}{1.5}
\section{Appendix}

\subsection{Notations}
The notations of the manuscript are summarized in  \tref{tab2}.
\begin{center}
\begin{longtable}{||l|p{4.5in}||}
\caption{Summary of notations.}
\label{tab2}
\endfirsthead
\endhead
\hline
\textit{Notation} & \textit{Description} \\ \hline
$N$ & An entangled quantum network, $N=\left(V,E\right)$, where $V$ is a set of nodes, $E$ is a set of entangled connections.  \\ \hline 
$A$  & A source user (quantum node) in the quantum network.  \\ \hline 
$B$  & A destination user (quantum node). \\ \hline 
$R_{j} $ & A current $j$-th quantum repeater, $j=1,\ldots ,q$, where $q$ is  the total number of quantum repeaters.  \\ \hline 
$R_{i} $ & A previous neighbor of $R_{i} $. \\ \hline 
$R_{k} $ & A next neighbor of $R_{k} $ (towards destination). \\ \hline 
$l$ & Level of entanglement. \\ \hline 
${\rm L}_{l} \left(x,y\right)$ & An $l$-level entangled connection between quantum nodes $x$ and $y$. \\ \hline 
$d\left(x,y\right)_{{\rm L}_{l} } $ & Hop-distance at an ${\rm L}_{l} $-level entangled connection between quantum nodes $x$ and $y$, $d\left(x,y\right)_{{\rm L}_{l} } =2^{l-1} $.  \\ \hline 
$O_{C} $ & An oscillator with frequency $f_{C} $, $f_{C} ={1\mathord{\left/ {\vphantom {1 t_{C} }} \right. \kern-\nulldelimiterspace} t_{C} } $, serves as a reference clock.   \\ \hline 
$C$ & A cycle, with $t_{C} ={1\mathord{\left/ {\vphantom {1 f_{C} }} \right. \kern-\nulldelimiterspace} f_{C} } $. \\ \hline 
$\pi _{S} $ & An entanglement swapping period, in which the set ${\rm {\mathcal S}}_{I} \left(R_{j} \right)$ of density matrices are swapped via the $U_{S} $ entanglement swapping operator with the ${\rm {\mathcal S}}_{O} \left(R_{j} \right)$ of density matrices, defined as $\pi _{S} =xt_{C} $, where $x$ is the number of $C$.   \\ \hline 
$\pi '_{S} $ & A next entanglement swapping period after $\pi _{S} $. \\ \hline 
$q$ & Total number of quantum repeaters in an entangled path ${\rm {\mathcal P}}\left(A_{i} \to B_{i} \right)$,  $q=d\left(A,B\right)_{{\rm L}_{l} } -1$. \\ \hline 
$B_{F} $ & Entanglement throughput [Bell states per $\pi _{S} $]. \\ \hline 
$\left|B_{F} \right|$ & Number of entangled states [Number of Bell states]. \\ \hline 
${\rm L}_{l} \left(k\right)$ & A $k$-th entangled connection. \\ \hline 
$B_{F} \left({\rm L}_{l} \left(k\right)\right)$ & Entanglement throughput of the entangled connection ${\rm L}_{l} \left(k\right)$ [Bell states per $\pi _{S} $].  \\ \hline 
$\rho $ & In a $j$-th quantum repeater $R_{j} $, an $\rho $ incoming density matrix is a half of an entangled state ${\left| \beta _{00}  \right\rangle} $ received from a previous neighbor node $R_{j-1} $. \\ \hline 
$\sigma $ & The $\sigma $ outgoing density matrix in $R_{j} $ is a half of an entangled state ${\left| \beta _{00}  \right\rangle} $ shared with a next neighbor node $R_{j+1} $. \\ \hline 
$U_{S} $ & Entanglement swapping operation is a local transformation that swaps an incoming density matrix $\rho $ with an outgoing density matrix $\sigma $ in a quantum repeater $R$. \\ \hline 
${\rm {\mathcal S}}_{I} \left(R_{j} \right)$ & Set of incoming density matrices stored in the quantum memory of $R_{j} $, ${\rm {\mathcal S}}_{I} \left(R_{j} \right)=\bigcup _{i}\rho _{i}  $, where $\rho _{i} $ is an $i$-th density matrix. \\ \hline 
${\rm {\mathcal S}}_{O} \left(R_{j} \right)$ & Set of outgoing density matrices stored in the quantum memory of $R_{j} $, ${\rm {\mathcal S}}_{O} \left(R_{j} \right)=\bigcup _{i}\sigma _{i}  $, where $\sigma _{i} $ is an $i$-th density matrix. \\ \hline 
${\rm {\mathcal S}}_{I}^{*} \left(R_{j} \right)$ & A complete set of incoming density matrices. Set ${\rm {\mathcal S}}_{I} \left(R_{j} \right)$ formulates a ${\rm {\mathcal S}}_{I}^{*} \left(R_{j} \right)$ complete set if ${\rm {\mathcal S}}_{I} \left(R_{j} \right)$ contains all the $Q=\sum _{i=1}^{N}\left|B_{i} \right| $ incoming density matrices per $\pi _{S} $ that is received by $R_{j} $ in a swapping period, where $N$ is the number of input entangled connections of $R_{j} $, $\left|B_{i} \right|$ is the number of incoming densities of the $i$-th input connection per $\pi _{S} $, thus ${\rm {\mathcal S}}_{I} \left(R_{j} \right)=\bigcup _{i=1}^{Q}\rho _{i}  $ and $\left|{\rm {\mathcal S}}_{I} \left(R_{j} \right)\right|=Q$. \\ \hline 
${\rm {\mathcal S}}_{O}^{*} \left(R_{j} \right)$ & A complete set of outgoing density matrices. An ${\rm {\mathcal S}}_{O} \left(R_{j} \right)$ set formulates a ${\rm {\mathcal S}}_{O}^{*} \left(R_{j} \right)$ complete set, if ${\rm {\mathcal S}}_{O} \left(R_{j} \right)$ contains all the $N$ outgoing density matrices that is shared by $R_{j} $ during a swapping period $\pi _{S} $, thus ${\rm {\mathcal S}}_{O} \left(R_{j} \right)=\bigcup _{i=1}^{N}\sigma _{i}  $ and $\left|{\rm {\mathcal S}}_{O} \left(R_{j} \right)\right|=N$. \\ \hline 
${\rm {\mathcal S}}\left(R_{j} \right)$ & An entanglement swapping set of $R_{j} $, ${\rm {\mathcal S}}\left(R_{j} \right)={\rm {\mathcal S}}_{I} \left(R_{j} \right)\bigcup {\rm {\mathcal S}}_{O} \left(R_{j} \right) $ that describes the status of the quantum memory in $R_{j} $. \\ \hline 
${\rm {\mathcal S}}^{*} \left(R_{j} \right)$ & A complete entanglement swapping set. A ${\rm {\mathcal S}}\left(R_{j} \right)$ is a ${\rm {\mathcal S}}^{*} \left(R_{j} \right)$ complete swapping set, if ${\rm {\mathcal S}}^{*} \left(R_{j} \right)={\rm {\mathcal S}}_{I}^{*} \left(R_{j} \right)\bigcup {\rm {\mathcal S}}_{O}^{*} \left(R_{j} \right) $, with cardinality $\left|{\rm {\mathcal S}}^{*} \left(R_{j} \right)\right|=Q+N$. \\ \hline 
${\rm {\mathcal S}}^{*} \left(R_{j} \right)$ & A perfect entanglement swapping set. A ${\rm {\mathcal S}}^{*} \left(R_{j} \right)$ complete swapping set is a $\hat{{\rm {\mathcal S}}}\left(R_{j} \right)=\hat{{\rm {\mathcal S}}}_{I} \left(R_{j} \right)\bigcup \hat{{\rm {\mathcal S}}}_{O} \left(R_{j} \right) $ perfect swapping set at a given $\pi _{S} $, if $\left|\hat{{\rm {\mathcal S}}}\left(R_{j} \right)\right|=N+N$. \\ \hline 
${\rm {\mathcal S}}_{R_{j} }^{\left(\pi _{S} \right)} \left(\left(R_{i} ,\sigma _{k} \right)\right)$ & A coincidence set, a subset of incoming density matrices in ${\rm {\mathcal S}}_{I} \left(R_{j} \right)$ of $R_{j} $ received from $R_{i} $ that requires the same outgoing density matrix $\sigma _{k} $ from ${\rm {\mathcal S}}_{O} \left(R_{j} \right)$ for the entanglement swapping.   \\ \hline 
$Z_{R_{j} }^{\left(\pi '_{S} \right)} \left(\left(R_{i} ,\sigma _{k} \right)\right)$ & Cardinality of the coincidence set ${\rm {\mathcal S}}_{R_{j} }^{\left(\pi _{S} \right)} \left(\left(R_{i} ,\sigma _{k} \right)\right)$ [Number of Bell states]. \\ \hline 
$\left|B\left(R_{i} \left(\pi _{S} \right),\sigma _{k} \right)\right|$ & Number of density matrices arrive from $R_{i} $ to $R_{j} $ for swapping with $\sigma _{k} $ at $\pi _{S} $. It increments the cardinality of the coincidence set as $Z_{R_{j} }^{\left(\pi '_{S} \right)} \left(\left(R_{i} ,\sigma _{k} \right)\right)=Z_{R_{j} }^{\left(\pi _{S} \right)} \left(\left(R_{i} ,\sigma _{k} \right)\right)+\left|B\left(R_{i} \left(\pi _{S} \right),\sigma _{k} \right)\right|$, where $\pi '_{S} $ is a next entanglement swapping period [Number of Bell states]. \\ \hline 
$\left|B_{R_{i} } \left(\pi _{S} \right)\right|$ & Incoming entanglement rate of $R_{j} $ per a $\pi _{S} $, defined as\newline $\left|B_{R_{j} } \left(\pi _{S} \right)\right|=\sum _{i,k}\left|B\left(R_{i} \left(\pi _{S} \right),\sigma _{k} \right)\right| $, where $\left|B\left(R_{i} \left(\pi _{S} \right),\sigma _{k} \right)\right|$ refer to the number of density matrices arrive from $R_{i} $ for swapping with $\sigma _{k} $ per $\pi _{S} $ [Bell states per $\pi _{S} $]. \\ \hline 
$D\left(\pi _{S} \right)$ & Delay, measured in entanglement swapping periods  $\pi _{S} $ [Number of $\pi _{S} $ periods]. \\ \hline 
$\left|B'_{R_{j} } \left(\pi _{S} \right)\right|$ & Outgoing entanglement rate of $R_{j} $, defined as\newline $\left|B'_{R_{j} } \left(\pi _{S} \right)\right|=\left(1-{\tfrac{L}{N}} \right){\tfrac{1}{1+D\left(\pi _{S} \right)}} \left(\left|B_{R_{j} } \left(\pi _{S} \right)\right|\right),$\newline where $L$ is the loss, $0<L\le N$ [Bell states per $\pi _{S} $].  \\ \hline 
$\zeta \left(\pi _{S} \right)$ & Entanglement swapping procedure at a given $\pi _{S} $. \\ \hline 
${\rm {\mathcal S}}_{I}^{\left(\pi _{S} \right)} \left(R_{j} \right)$ & Set of incoming densities of $R_{j} $ at $\pi _{S} $. \\ \hline 
$\left|{\rm {\mathcal S}}_{I}^{\left(\pi _{S} \right)} \left(R_{j} \right)\right|$ & Cardinality of the set ${\rm {\mathcal S}}_{I}^{\left(\pi _{S} \right)} \left(R_{j} \right)$ [Number of Bell states]. \\ \hline 
$\gamma $ & Noise coefficient, models the noise of the local quantum memory and the local operations, $0\le \gamma \le 1$. \\ \hline 
$N$ & Number of coincidence sets of $R_{j} $, and number of outgoing connections of $R_{j} $. \\ \hline 
$L$ & Number of losses, $0\le L\le N$. \\ \hline 
$M$ & Reduced number of swapped incoming and outgoing density matrices per $\pi _{S} $ at $L$ losses, $M=N-L$. \\ \hline 
$\gamma \left(\pi _{S} \right)$ & Noise at a given $\pi _{S} $. \\ \hline 
$Z_{R_{j} } \left(\pi _{S} \right)$ & A matrix of all coincidence set cardinalities for all input and output connections at $\pi _{S} $, defined as $Z_{R_{j} } \left(\pi _{S} \right)=Z_{R_{j} }^{\left(\pi _{S} \right)} \left(\left(R_{i} ,\sigma _{k} \right)\right)_{i\le N,k\le N} $. \\ \hline 
$\omega \left(\pi _{S} \right)$ & Weight coefficient, for a given entanglement swapping $\zeta \left(\pi _{S} \right)$ at a given $\pi _{S} $ is as\newline 
$\omega \left( {{\pi }_{S}} \right)=\sum\limits_{i,k}{{{\zeta }_{ik}}\left( {{\rho }_{A}},{{\sigma }_{k}} \right)Z_{{{R}_{j}}}^{\left( {{\pi }_{S}} \right)}\left( \left( {{R}_{i}},{{\sigma }_{k}} \right) \right)}=\left\langle \zeta \left( {{\pi }_{S}} \right),{{Z}_{{{R}_{j}}}}\left( {{\pi }_{S}} \right) \right\rangle ,$ \newline where $\left\langle \cdot \right\rangle $ is the inner product. \\ \hline 
$\omega ^{*} \left(\pi _{S} \right)$ & Maximized weight coefficient. \\ \hline 
$\zeta ^{*} \left(\pi _{S} \right)$ & Optimal entanglement swapping method at $\gamma \left(\pi _{S} \right)=0$ \\ \hline 
$\left|\chi \left(\pi _{S} \right)\right|$ & Norm, defined for an entanglement swapping $\chi \left(\pi _{S} \right)$. \\ \hline 
${\rm {\mathcal L}}\left(Z_{R_{j} } \left(\pi _{S} \right)\right)$ & Lyapunov function of $Z_{R_{j} } \left(\pi _{S} \right)$, as\newline ${\rm {\mathcal L}}\left(Z_{R_{j} } \left(\pi _{S} \right)\right)=\sum _{i,k}\left(Z_{R_{j} }^{\left(\pi _{S} \right)} \left(\left(R_{i} ,\sigma _{k} \right)\right)\right)^{2}  $. \\ \hline 
$C_{1} $, $C_{2} $ & Constants, $C_{1} >0$, $C_{2} >0$. \\ \hline 
$\Delta _{{\rm {\mathcal L}}} $ & Difference of Lyapunov functions ${\rm {\mathcal L}}\left(Z_{R_{j} } \left(\pi '_{S} \right)\right)$ and ${\rm {\mathcal L}}\left(Z_{R_{j} } \left(\pi _{S} \right)\right)$, where $\pi '_{S} $ is a next entanglement swapping period, defined as $\Delta _{{\rm {\mathcal L}}} ={\rm {\mathcal L}}\left(Z_{R_{j} } \left(\pi '_{S} \right)\right)-{\rm {\mathcal L}}\left(Z_{R_{j} } \left(\pi _{S} \right)\right)$. \\ \hline 
$\left|\bar{B}\left(R_{i} \left(\pi '_{S} \right),\sigma _{k} \right)\right|$ & A normalized number of arrival density matrices, $\left|\bar{B}\left(R_{i} \left(\pi '_{S} \right),\sigma _{k} \right)\right|\le 1,$ from $R_{i} $ for swapping with $\sigma _{k} $ at a next entanglement swapping period $\pi '_{S} $, defined as\newline $\left|\bar{B}\left(R_{i} \left(\pi '_{S} \right),\sigma _{k} \right)\right|={\tfrac{\left|B\left(R_{i} \left(\pi '_{S} \right),\sigma _{k} \right)\right|}{\left|B_{R_{j} } \left(\pi _{S} \right)\right|}} $,\newline where $\left|B_{R_{j} } \left(\pi _{S} \right)\right|=\sum _{i,k}\left|B\left(R_{i} \left(\pi _{S} \right),\sigma _{k} \right)\right| $ is a total number of incoming density matrices of $R_{j} $ from the $N$ quantum repeaters. \\ \hline 
${\mathbb{E}}\left(\left|\bar{B}\left(R_{i} \left(\pi _{S} \right),\sigma _{k} \right)\right|\right)$ & Expected normalized number of density matrices arrive from $R_{i} $ for swapping with $\sigma _{k} $ at $\pi _{S} $. \\ \hline 
$\alpha _{ik} $ & Parameter, defined as $\alpha _{ik} =Z_{R_{j} }^{\left(\pi _{S} \right)} \left(\left(R_{i} ,\sigma _{k} \right)\right)\bar{B}\left(R_{i} \left(\pi _{S} \right),\sigma _{k} \right)$. \\ \hline 
$\nu _{z} $ & A constant, $\nu _{z} \ge 0$. \\ \hline 
$C_{1} $ & Constant, $C_{1} =1-\sum _{z}\nu _{z}  $. \\ \hline 
$\left|Z_{R_{j} } \left(\pi _{S} \right)\right|$ & Cardinality of the coincidence sets at a given $\pi _{S} $, as\newline $\left|Z_{R_{j} } \left(\pi _{S} \right)\right|=\sum _{i,k}Z_{R_{j} }^{\left(\pi _{S} \right)} \left(\left(R_{i} ,\sigma _{k} \right)\right) =\left|{\rm {\mathcal S}}_{I} \left(R_{j} \right)\right|$. \\ \hline 
$\left|B_{R_{i} } \left(\pi _{S} \right)\right|$ & Total number of incoming density matrices in $R_{j} $ per a given $\pi _{S} $, as $\left|B_{R_{j} } \left(\pi _{S} \right)\right|=\sum _{i,k}\left|B\left(R_{i} \left(\pi _{S} \right),\sigma _{k} \right)\right| $. \\ \hline 
$\tilde{\pi }_{S} $ & An extended entanglement swapping period, defined as $\tilde{\pi }_{S} =\pi _{S} +D\left(\pi _{S} \right)$,\newline with ${\pi _{S} \mathord{\left/ {\vphantom {\pi _{S}  \left(\tilde{\pi }_{S} \right)}} \right. \kern-\nulldelimiterspace} \left(\tilde{\pi }_{S} \right)} \le 1$ [Number of $\pi _{S} $ periods]. \\ \hline 
$B'_{R_{j} } \left(\pi _{S} \right)$ & Outgoing entanglement rate per $\pi _{S} $ for a particular entanglement swapping set [Bell states per $\pi _{S} $]. \\ \hline 
$f\left(\cdot \right)$ & Sub-linear function. \\ \hline 
$\xi \left(\gamma \right)$ & Parameter, defined as $\xi \left(\gamma \right)={\tfrac{\left(N-L\right)}{2C_{1} }} f\left(\gamma \left(\pi _{S} \right)\right)$. \\ \hline 
$\beta $ & Parameter, defined as\newline $\beta =\sum _{i,k}\left(\left|\bar{B}\left(R_{i} \left(\pi _{S} \right),\sigma _{k} \right)\right|-\left|\bar{B}\left(R_{i} \left(\pi _{S} \right),\sigma _{k} \right)\right|^{2} \right) $, \newline where $\left|\bar{B}\left(R_{i} \left(\pi _{S} \right),\sigma _{k} \right)\right|$ refers to the normalized number of density matrices arrive from $R_{i} $ for swapping with $\sigma _{k} $ at $\pi _{S} $ as\newline $\left|\bar{B}\left(R_{i} \left(\pi _{S} \right),\sigma _{k} \right)\right|={\tfrac{\left|B\left(R_{i} \left(\pi _{S} \right),\sigma _{k} \right)\right|}{\sum _{i}\left|B\left(R_{i} \left(\pi _{S} \right),\sigma _{k} \right)\right| }} $. \\ \hline 
$P$ & Number of entanglement swapping periods . \\ \hline 
$D\left(\pi _{S} \right)$ & Delay per $\pi _{S} $ at a non-complete entanglement swapping set. \\ \hline 
$D^{*} \left(\pi _{S} \right)$ & Delay per $\pi _{S} $ at a complete entanglement swapping set. \\ \hline 
$\hat{D}\left(\pi _{S} \right)$ & Delay per $\pi _{S} $ at a perfect entanglement swapping set. \\ \hline 
$\left|Z_{R_{j} } \left(\pi _{S} \right)\right|$ & Cardinality of the coincidence sets at a given $\pi _{S} $, for a non-complete entanglement swapping set. \\ \hline 
$\left|Z_{R_{j} }^{*} \left(\pi _{S} \right)\right|$ & Cardinality of the coincidence sets at a given $\pi _{S} $, for a complete entanglement swapping set. \\ \hline 
$\left|\hat{Z}_{R_{j} } \left(\pi _{S} \right)\right|$ & Cardinality of the coincidence sets at a given $\pi _{S} $, for a perfect entanglement swapping set. \\ \hline 
$\pi _{S}^{*} $ & An extended entanglement swapping period, defined as $\pi _{S}^{*} =\left(1+h\right)\pi _{S} $, where $h>0$ [Number of $\pi _{S} $ periods].  \\ \hline
\end{longtable}
\end{center}
\end{document}